\documentclass[12pt]{amsart}   	
\usepackage{graphicx}
\usepackage{color}

\newtheorem{theorem}{Theorem} 
\newtheorem{conjecture}[theorem]{Conjecture}
\newtheorem{lemma}[theorem]{Lemma}

\newtheorem{proposition}[theorem]{Proposition}

\title[Equitable Retirement Income Tontines]{Equitable Retirement Income Tontines: \\ Mixing Cohorts Without Discriminating}
\author{M.A. Milevsky and T.S. Salisbury}
\thanks{Milevsky is an Associate Professor of Finance at the Schulich School of Business, York University, and Executive Director of the IFID Centre. Salisbury is a Professor in the Department of Mathematics and Statistics at York University. The authors acknowledge funding from the IFID Centre (Milevsky) and from NSERC (Salisbury) as well as helpful comments from two ASTIN reviewers, the editor (Daniel Bauer) as well as seminar participants at University of New South Wales, Macquarie University, the CFA Society of Sydney and in particular Anthony Asher, David Blake, Steven Haberman, Geoff Kingston, Sachi Purcal, John Piggott and Michael Sherris. The contact author (Milevsky) can be reached at: milevsky@yorku.ca} 
\date{Final Accepted Version: April 14, 2016}

\begin{document} 
\begin{abstract}

There is growing interest in the design of pension annuities that insure against {\em idiosyncratic} longevity risk while pooling and sharing {\em systematic} risk. This is partially motivated by the desire to reduce capital and reserve requirements while retaining the value of mortality credits; see for example Piggott, Valdez and Detzel (2005) or Donnelly, Guillen and Nielsen (2014).  In this paper we generalize the {\em natural retirement income tontine} introduced by Milevsky and Salisbury (2015) by combining heterogeneous cohorts into one pool. We engineer this scheme by allocating tontine shares at either a premium or a discount to par based on both the age of the investor and the amount they invest. For example, a 55 year-old allocating $\$10,000$ to the tontine might be told to pay $\$$200 per share and receive 50 shares, while a 75 year-old allocating $\$8,000$ might pay $\$$40 per share and receive 200 shares. They would all be mixed together into the same tontine pool and each tontine share would have equal income rights. The current paper addresses existence and uniqueness issues and discusses the conditions under which this scheme can be constructed {\em equitably} -- which is distinct from {\em fairly} -- even though it isn't {\em optimal} for any cohort. As such, this also gives us the opportunity to compare and contrast various pooling schemes that have been proposed in the literature and to differentiate between arrangements that are socially equitable, vs. actuarially fair vs. economically optimal.

\end{abstract}

\maketitle
%\tableofcontents

\newpage

\section{Introduction}
\label{sec:intro}

The tontine annuity -- which was first promoted as a retirement income vehicle by Lorenzo di Tonti in the year 1653 -- hasn't benefited from the best publicity over the last three and a half centuries. Although at first tontines were used by British and French governments to finance their wars (against each other), conventional fixed interest bonds ended up superseding them as the preferred method of deficit financing. Private sector insurance companies in the 18th century offered tontine-like products, but they too were superseded by more familiar guaranteed life annuities and pensions. In fact, by the early 19th century regulators in the U.S. and the U.K. banned (a derivative product called) tontine insurance, although there is some debate over whether the ban actually applies to the tontines envisioned by Tonti. Legalities aside, in the words of the well-known financial writer Edward Chancellor, tontines are {\em ``one of the most discredited financial instruments in history.''} We refer the interested reader to the book by Milevsky (2015) in which a slice of the tontine's colorful history is addressed. In this article our focus (and contribution) is actuarial as opposed to political or historical.

In its purest financial form a tontine annuity can be viewed as a perpetual (i.e. infinite maturity) bond that is purchased from an issuer by a group of investors who agree to share periodic coupons only amongst survivors.  As investors die and leave the tontine pool the coupons or cash flows earned by those who avoid death increase (super) exponentially over time. In theory the last remaining survivor receives all of the coupons until he or she finally dies and the issuer's obligations to make payments are terminated. One can alternatively think of the tontine annuity as consisting of a portfolio of zero coupon bonds (ZCBs) with staggered maturities or face values in which the final ZCB matures at the maximum possible lifespan of the investors in the pool, e.g. age 125. Cash flows from maturing ZCBs are distributed equally among survivors. From this perspective the cash-flow pattern can be fined-tuned to any desired profile as long as its present value is equal to the amount invested by the group.  The tontine pool retains longevity risk in the sense that if people live longer than expected their payments are reduced relative to what they might have expected at time zero. Like other pooled schemes we will discuss -- such as those described in Piggott et. al. (2005), Stamos (2008), or Donnelly et. al. (2014) -- with a {\em retirement income tontine} there is no entity guaranteeing fixed payments for life, thus eliminating capital requirements. We use the term {\em retirement income tontine} to differentiate our scheme from a ``winner take all'' bet in which the payoff is deferred to the last survivor and to remind the reader of the pension-like structure.

\subsection{Problems with Tontines}
\label{subsec:problems}

Over the past few centuries there has been quite a bit of popular and scholarly criticism leveled against tontine schemes. Overall these concerns can be placed into three broad categories. 

The first concern is that tontines themselves are amoral because one immediately benefits from someone's death. A more refined version of this concern is that they create an incentive for fraud, murder and other criminal activity. These critics contend that as the size of the tontine pool shrinks the surviving members are incentivized to kill each other to gain a bigger share of the tontine pool. This colorful perception of the tontine permeates literary fiction and is the subject of many novels, but has no basis in reality. Most historical tontine schemes capped or limited payouts once a small fraction (say 5 to 10 members) remained in the pool. Despite all the fictional novels there simply is no documented evidence that the last few survivors of a national tontine ever murdered each other. In fact, with hundreds of people in the tontine pool the economic benefit from nefariously reducing the pool size is minimal. The ethical concern that investors would benefit directly from death can be dismissed outright within the actuarial community since that is the foundation of all pension and annuity pricing. As for the concern with fraud, this indeed was a problem in the 17th and 18th century when documenting life and death was unreliable but can also be dismissed in the 21st century with modern record-keeping systems. 

The second concern with tontine annuities relates to economic optimality and cash flow patterns. In Lorenzo's tontine scheme the cash paid to the group remained relatively constant (i.e. the numerator) but the number of survivors (i.e. the denominator) declined more-than exponentially fast over time, resulting in a rapidly increasing payout to surviving members. This explosive profile of income is at odds with the economic desire for (stable) consumption smoothing. It is assumed that older retired investors would want a stable or perhaps even declining real cash-flow over time, notwithstanding concerns about health-care expenses and inflation for the elderly. According to these critics the tontine annuity isn't an optimal economic contract. But this concern can also be remedied with proper product design. There is no reason why the cash-flows to the group should be structured to remain constant over time. As mentioned above, the ZCBs payments to the pool could decline (faster than) exponentially at the same rate as (expected) mortality. In fact, this is the essence of Milevsky and Salisbury (2015). To sum up, we believe the second objection can be easily overcome.\footnote{This argument was recently made by Professor William Sharpe in a presentation to the French Finance Association, quoted in Milevsky (2015, pg. 164).}

The third concern is a more-subtle one and has to do with the pooling of cohorts and the age profile of the tontine. That is the core focus on this paper. If one allows anyone regardless of age to participate equally in the same tontine annuity pool -- which we do not -- there would be an immediate transfer of wealth from the members who are expected to die early (i.e. the old) to those who are expected to live the longest (i.e. the young). Historical tontines -- such as the one first issued by the English government in the late 17th century -- discriminated against the old in favor of the young. Thus, for example, in the earliest English tontine schemes the nominees on whose life the tontine was contingent ranged in age from a few months to over the age of 50. In equilibrium everyone should nominate the healthiest possible age (females age 10, approximately), but that leads to design problems when the nominee and annuitant are not the same person. If indeed one requires homogenous mortality pools to run a non-discriminating tontine scheme then this limits the possible size of the pool and the efficacy of large number diversification. The 18th century tontine schemes in which investors were placed into small tailored classes -- with different payouts based on age -- suffered from reduced pooling and risk diversification. A modern day pension fund trying to implement a tontine payout structure with a predictable cash flow would face the same concern unless it had a very large pool of willing retirees with identical ages and health profiles. 

In some sense, we believe this is the most serious criticism against resurrecting retirement income tontines in the context of modern pension schemes. One would require hundreds of people of exactly the same age retiring on exactly the same day, to reduce the variability of payouts. In fact, this is one reason why authors such as Piggot, Valdez and Detzel (2005) or Donnelly, Guillen and Nielsen (2014) have proposed pooled annuitization schemes or overlays that allow for mixing of different cohorts over multiple generations. Note that we will not wade into a debate over which among the many pooling schemes is {\em better.} In fact there is quite a bit of overlap between them as we will soon demonstrate. Moreover there will be a tradeoff. A scheme that squeezes out the highest possible utility for the group may also be more complex to analyze and harder to explain. 

In sum, we choose to analyze tontines which offer a design that provides ``good'' utility while remaining sufficiently simple that we can establish a range of qualitative results. We understand and acknowledge that other designs have their own appeal and role.

\subsection{Making heterogeneous tontines equitable}
\label{subsec:resolution}

As stated above, it is this third criticism leveled against (retirement income) tontines that we address in the current paper. Our remedy to this concern is to allow cohorts of different ages (and mortality) to mix in the same pool by allocating different participation rates or shares based on their age at the time of purchase. For example, a 55 year-old allocating $\$10,000$ to the tontine might be told to pay \$200 per share and receive 50 shares, while a 75 year-old allocating $\$8,000$ might pay $\$$40 per share and receive 200 shares. They would all be mixed together into the same tontine pool and each tontine share would have equal income rights.\footnote{On a historical note, this proposal to ``fix'' tontines was actually made almost 200 years ago by Mr. Charles Compton (1833), who was the Accountant General of the Royal Mail in the UK. He wrote:  {\em ``His Majesty's Government should create a Tontine Stock bearing interest at a certain rate per annum and to permit persons of a certain age to purchase such stock at par, those younger or older to purchase the same stock above or below par according to their several ages.''} He argued that this was preferable to forming the contributors into classes which ends-up creating very small groups with few benefits from pooling. He claimed that: {\em ``Younger purchasers would give more money for \pounds 100 stock than the elder and would give up part of their income for the benefit of elder members, who in turn would bequeath their annuities to the younger as compensation."}  Compton (1833) goes on to list a table of values mapping ages into share prices, which is quite similar (in spirit, at least) to our numbers, which we present later on in the paper.}

Now, it might seem rather trivial (actuarially) to allocate shares in the tontine based on the age of the investor and the size of their investment. After all, with an immediate annuity, \$1 of lifetime income will cost $a_{45}$ for a 45-year old and $a_{75}$ for a 75 year-old. The relative prices of mortality-contingent claims are well understood in the actuarial literature. But what may not be obvious is that in fact there are situations (i.e. counter examples) in which this cannot be done in a fair (or even equitable) manner, especially when the groups are small. In other words, there are cases in which no mapping or share price will allow groups to be mixed without discrimination. Our objective is to understand when this is (or is not) possible. The need for large pools to diversify risk is linked to the issues addressed in this paper and is a question that has recently been highlighted by Donnelly (2015) as well. We return to this later.

As far as terminology is concerned, in this paper we are careful to distinguish between a scheme that is {\em fair} and a scheme that is {\em equitable}, which is a somewhat weaker requirement. A retirement income tontine scheme in which there is a possibility of everyone in the pool dying before the maximum age and thus leaving left-overs can never be made {\em fair}, in the sense of Donnelly (2015) unless it incorporates some form of payment to estates. By the word {\em fair} we mean that the expected present value of income will always be less than the amount contributed or invested into the tontine. However, a heterogeneous tontine scheme can often (though not always) be made {\em equitable} by ensuring that the present value of income (although less than the amount contributed) is the same for all participants in the scheme regardless of age. This scheme will not discriminate against any one cohort although it won't be fair. All of this will be addressed in detail including an analysis of scenarios in which equity is impossible to achieve.

To recap then, in this paper we investigate how to construct a multi-age tontine scheme  and determine the proper share prices to charge participants so that it is equitable and doesn't discriminate against any age or any group. The tontine we propose is a closed pool that does not allow anyone to enter (or obviously exit) after the initial set-up. This is one further place we differ from the designs of Piggott et. al. (2005), Donnelly et. al. (2014), or (in the tontine context) Sabin (2010). That is, we advocate closing the group to newcomers, but allow multiple ages and contribution levels within the closed pool. Again, we refrain from arguing that this is better or worse than any other design. That said, our model requires and assumes little (if any) {\em actuarial discretion} as time evolves. The rules are set at time zero and the cash-distribution algorithm is crystal clear. We believe this design has merit.

\subsection{Outline}
\label{subsec:outline}

The remainder of this paper is organized as follows. In Section \ref{sec:background} we summarize the results of  Milevsky and Salisbury (2015) and the economic optimality of tontines in the context of a single homogeneous cohort of subscribers. Indeed, there are many possible payout functions $d(t)$ that one can use to construct a tontine scheme -- historically the $d(t)$ curve was constant, for example 8\% per year  -- and the optimal function depends (at a minimum) on the representative investor's coefficient of risk aversion. A particularly natural payout function arises as the $d(t)$ that optimizes logarithmic utility preferences. In this section we also offer a brief comparison to other pooling schemes. Section \ref{sec:mixing} moves from a review of single cohorts to the introduction of multiple cohorts, which is the contribution of the current paper. It describes in precise terms what is meant by an equitable share price for all participants in a tontine scheme given a particular payout function $d(t)$. In that section we offer a more precise definition of the notion of fairness and how it differs from equitable. Section \ref{sec:utility} returns to the matter of economic optimality. It is typically impossible to locate a payout function $d(t)$ that is optimal for all cohorts, as was possible  in the case of a single cohort, even if all participants have the same level of risk aversion. Indeed, the best that one can hope for when mixing cohorts is an equitable scheme and not a uniformly optimal one. In Section \ref{sec:natural} we propose that a good selection from among all the possible equitable schemes is one in which the payout function $d(t)$ is proportional to the expected number of shares outstanding at any point in the future. Alas, we can't yet prove uniqueness for this scheme and leave this as a conjecture. We do however discuss welfare gains and losses from the scheme and provide some numerical examples. Section \ref{sec:others} makes comparisons with other product designs that exist in the literature and discusses conditions under which they overlap. Section \ref{sec:conclusion} concludes and offers some suggestions and avenues for further research. Proofs appear in the appendix (Section \ref{sec:appendix}). 

\section{Annuities vs. Optimal Tontine Payout Functions}
\label{sec:background}

In this section we briefly review the {\em optimal} tontine scheme proposed in Milevsky and Salisbury (2015). We assume an objective survival function ${}_tp_x$, for an individual aged $x$ to survive $t$ years. One purpose of the tontine structure is to insulate the issuer from the risk of a stochastic (or simply mis-specified) survival function, but in this paper we assume ${}_tp_x$ is given and applies to all individuals. We intend to address the stochastic case in subsequent work. We assume that the tontine pays out continuously as opposed to quarterly or monthly. For ease of exposition we assume a constant risk-free interest rate $r$, though it would be easy to incorporate a term structure. What makes the tontine a simple and inexpensive product to build and manage is that the payouts are known from the beginning and can be engineered (without active management) by a simple portfolio of ZCBs. 

\subsection{Optimal annuities}
\label{subsec:annuities}

The basic comparator for a tontine is an annuity in which annuitants each pay \$1 to the insurer initially and receive in return an income stream of $c(t)\,dt$ for life. The constraint on these annuities is that they are fairly priced, in other words that with a sufficiently large client base the initial payments invested at the risk-free rate will fund the called-for payments in perpetuity. This implies a constraint on the annuity payout function $c(t)$, namely that 
\begin{equation}
\label{annuityconstraint}
\int_0^\infty e^{-rt}{}_tp_x\, c(t)\,dt=1.
\end{equation}
Again, $c(t)$ is the payout rate per survivor. The payout rate per initial dollar invested is ${}_tp_x\,c(t)$. Letting $U(c)$ denote the instantaneous utility of consumption, a rational annuitant (with lifetime $\zeta$) having no bequest motive will choose a life annuity payout function for which $c(t)$ maximizes the discounted lifetime utility:
\begin{equation*}
E[\int_0^\zeta e^{-rt}U(c(t))\,dt]=\int_0^\infty e^{-rt}{}_tp_x\, U(c(t))\,dt
%\label{annuityutilityspecification}
\end{equation*}
where $r$ is (also) the subjective discount rate (SDR), all subject to the constraint \eqref{annuityconstraint}. Provided $u$ is strictly concave, the following now follows from the Euler-Lagrange theorem.
\begin{theorem}[Theorem 1 of Milevsky and Salisbury (2015)] Optimized life annuities have constant $c(t)\equiv \frac{1}{a_x}$, where
$a_x=\int_0^\infty e^{-rt}{}_tp_x\,dt$.
\label{thm:annuity}
\end{theorem}
This result can be traced back to Yaari (1965) who showed that the optimal (retirement) consumption profile is constant (flat) and that 100\% of wealth is annuitized when there is no bequest motive, subjective discount rates are equal to interest rates and complete annuity markets (actuarial notes) are available. 

\subsection{Optimal Tontine Payouts}
\label{subsec:payout}

In practice, insurance companies funding the life annuity $c(t)$ are exposed to both systematic longevity risk (due to randomness or uncertainty in ${}_tp_x$), model risk (the risk that ${}_tp_x$ is mis-specified), as well as re-investment or interest rate risk -- a static bond portfolio can replicate the tontine payout over the medium term, but not beyond 30 years. The latter is not our focus here so we will continue to assume that $r$ is a given constant for most of what follows. 

This brings us to the tontine structure introduced in Milevsky and Salisbury (2015), in which a predetermined dollar amount is shared among survivors at every $t$. Let $d(t)$ be the rate at which funds are paid out per initial dollar invested, a.k.a. the tontine payout function. There is no reason for the tontine payout function to be a constant fixed percentage of the initial dollar invested (e.g. 4\% or 7\%) as it was historically. Getting back to the issue of optimality, we can pose the same question as considered above for annuities: what $d(t)$ is best for subscribers? The comparison is now between $d(t)$ and ${}_tp_x\,c(t)$, where $c(t)$ is the optimal annuity payout found above. 

Suppose there are initially $n$ subscribers to the tontine scheme, each depositing a dollar with the tontine sponsor. Let $N(t)$ be the random number of live subscribers at time $t$. Consider one of these subscribers. Given that this individual is alive, $N(t)-1\sim \text{Bin}(n-1,{}_tp_x)$. In other words, the number of other (live) subscribers at any time $t$ is binomially distributed with probability parameter ${}_tp_x$. 

As in the Yaari (1965) model, this individual's discounted lifetime utility is
\begin{align*}
&E[\int_0^\zeta e^{-rt }u\Big(\frac{n d(t)}{N(t)}\Big)\,dt]=\int_0^\infty e^{-rt}{}_tp_x \,E[u\Big(\frac{n d(t)}{N(t)}\Big)\mid \zeta>t]\,dt\\
&\qquad=\int_0^\infty e^{-rt}{}_tp_x\sum_{k=0}^{n-1} \binom{n-1}{k}{}_tp_x^k(1-{}_tp_x)^{n-1-k}u\Big(\frac{nd(t)}{k+1}\Big)\,dt.
\end{align*}
The constraint on the tontine payout function $d(t)$ is that the initial deposit of $n$ should be sufficient to sustain withdrawals in perpetuity. Of course, at some point all subscribers will have died. So in fact the tontine sponsor will eventually be able to cease making payments, leaving a small remainder or windfall. This gets to the issue of {\em fairness}, which we revisit in the next section. But this final-death time is not predetermined, so we treat that profit as an unavoidable feature of the tontine. Remember that we do not want to expose the sponsor to any longevity risk. It is the pool that bears this risk entirely. 

Our budget or pricing constraint is therefore that 
\begin{equation}
\label{tontineconstraint}
\int_0^\infty e^{-rt} d(t)\,dt=1.
\end{equation}

So, for example, if $d(t)=d_0$ is forced to be constant, i.e. a \emph{flat tontine} as was typical historically, then the tontine payout function (rate) is simply $d_0=r$ (or somewhat more if a cap on permissible ages is imposed, replacing the upper bound of integration in \eqref{tontineconstraint} by a value less than infinity).  

The \emph{optimal} $d(t)$ is in fact far from constant. Milevsky and Salisbury (2015) find this optimum in some generality. The following summarizes the conclusion in the case of CRRA utility $U(c)=\frac{c^{1-\gamma}}{1-\gamma}$ for $\gamma>0$, $\gamma\neq 1$ (or $U(c)=\log c$ for $\gamma=1$). Set 
$$
\beta_{n,\gamma}(p)=p\sum_{k=0}^{n-1}\binom{n-1}{k}p^k(1-p)^{n-1-k}\Big(\frac{n}{k+1}\Big)^{1-\gamma}.
$$

\begin{theorem}[Corollary 2 of Milevsky and Salisbury (2015)] The optimal retirement income tontine structure has $d(t)=d(0)\beta_{n,\gamma}({}_tp_x)^{1/\gamma}$, where $d(0)$ is chosen to make \eqref{tontineconstraint} hold.
\label{thm:cohortoptimal}
\end{theorem}
For an illustration of the typical such $d(t)$, see Figure \ref{fig03}. 

\subsection{Comparing to Others}

Tontines are not the only alternative to annuities that have been investigated in the actuarial literature. Two such product designs are the {\it group self annuitization scheme} (GSA) of Piggott, Valdez and Detzel (2005), and the optimal {\it pooled annuity fund} (PAF) analyzed by Stamos (2008). We will describe these designs in Section \ref{sec:others}. But for those already familiar with them we make some comparisons now, that may help set the context for our results. 

First of all, there is some degree of overlap between the three designs. In fact, for a homogeneous pool as described above, invested in risk-free assets, it turns out that a GSA agrees with a tontine having payout $d(t)=\frac{1}{a_x}{}_tp_x$. Following Milevsky and Salisbury (2015) we call this design a {\it natural tontine for the age-$x$ cohort}. That paper showed that this design is optimal for logarithmic utility ($\gamma=1$), since $\beta_{n,1}(p)=p$. For heterogeneous pools however, the GSA will not be a tontine at all, since the total payout will be random (and path dependent), rather than deterministic (which is the defining feature of a tontine). In the same context it turns out that the PAF which is optimal for logarithmic utility also agrees with a natural tontine. For risk aversion $\gamma\neq 1$ in contrast, the total PAF payout will be path-dependent, so again it is not a tontine. Recall that the main defining feature of the tontine is the predictability of cash-flows (numerator) distributed to the pool.

We will make these claims precise in Section \ref{sec:others}, but at this point we will simply content ourselves with Table \ref{table09}. 
$$
\text{\bf Insert Table \ref{table09} here}
$$
It compares the certainty equivalent of investments in all three products yielding the same utility as \$100 in an annuity. As must be the case, a (fairly priced) annuity provides the highest utility, followed by the PAF, then the tontine and then the GSA. Moreover the three agree when $\gamma=1$. But the principal conclusion is that the three designs yield utilities that are very similar, even with a (very) small pool of investors. If the annuity is not fairly priced, e.g. if capital risk charges are imposed upon an annuity, in the form of a loading factor that protects against systematic mortality risk, then any of the three designs can easily provide higher utility to the consumer than the annuity. 

Note that in Table \ref{table09} we do not include the GSA scheme for $\gamma=2$ or 5, because in fact it has mean utility $=-\infty$ once $\gamma>2$. This is an artifact of taking an infinite horizon and in the context of a natural tontine is discussed at greater length in section A.3 of the Appendix of Milevsky and Salisbury (2015). As indicated there, the disproportionate influence of extreme ages could be circumvented by capping payments at some advanced age such as 110. 

\section{Mixing cohorts: Equitable share prices when $d(t)$ is given}
\label{sec:mixing}

Now suppose a retirement income tontine pays out $d(t)$ per initial dollar invested, which may (or may not) be optimal for people of a given age. In other words, an inhomogeneous group of individuals subscribe to purchase shares in the tontine. Each subscriber will be entitled to a share of the total funds disbursed in proportion to the number of shares owned, with the sole caveat that the subscriber must be alive at the time of disbursement.  Once the list of subscribers is known, together with the dollar value they will invest, they are each quoted a price per share depending on their age and the size of their investment (and more generally, the ages and investments of all subscribers). Of course, the price then determines the number of shares they will receive in return for their announced investment. The issue we wish to address is how to assign prices in an equitable manner which we will take to mean that the expected present value of funds received by the various subscribers (per initial dollar invested) are all equal. The mathematical question becomes whether there exists a collection of share prices that realize this and whether such prices are unique. 

Let $n$ be the number of subscribers. For computational purposes it will sometimes be convenient to group them into $K$ homogeneous cohorts (i.e. with the same age and contribution level), though this is not actually a restriction since we could choose to take $K=n$ and deem each cohort to consist of a single individual. We will use notation that permits grouping, but in many proofs will (without loss of generality) take cohorts to consist of single individuals. For $i=1,\dots,K$ let $x_i$ be the initial age of individuals in the $i$th cohort, and let $w_i$ be the number of dollars each of them invests. Let $n_i$ be the size of the $i$th cohort, so $n=\sum n_i$ and the total initial investment is $w=\sum n_iw_i$. Therefore the total time-$t$ payouts occur at rate $wd(t)$. For notational convenience, we choose to base prices on {\it participation rates} $\pi_i$, in other words, $1/\pi_i$ is the $i$th subscriber's price per share. Let $u_i=\pi_iw_i$ be the number of shares purchased by each individual in the $i$th cohort, and let $u=\sum n_iu_i$ be the total number of shares purchased. Let $N_i(t)$ be the number in the $i$th cohort who survive to time $t$. The $d(t)$ function satisfies the budget constraint $\int_0^\infty e^{-rt}d(t)\,dt=1$, where $r$ is the interest rate. An individual in the $i$th cohort who survives to time $t$ will receive payments at a rate of 
$$
u_i\times\frac{wd(t)}{\sum_j u_j N_j(t)} = w d(t)\frac{\pi_i w_i}{\sum_j \pi_jw_jN_j(t)}.
$$
Summing this over all subscribers of course gives back a total payout rate of $wd(t)$, as long as at least one subscriber survives. In other words, $\sum_iwd(t) N_i(t)\frac{\pi_iw_i}{\sum_j \pi_j w_jN_j(t)}=wd(t)$, as long as $\sum_j \pi_j w_jN_j(t)>0$. 

Let $\zeta$ be the lifetime of an individual in the $i$th cohort. The present value of their payments is
\begin{align*}
E\Big[\int_0^{\zeta} e^{-rt}u_i\frac{wd(t)}{\sum_j u_j N_j(t)}\,dt\Big]
&=\int_0^\infty e^{-rt}{}_tp_{x_i}\pi_i w_iE\Big[\frac{wd(t)}{\sum_{j} \pi_j w_jN_j(t)}\mid \zeta>0\Big]\,dt\\
&=\int_0^\infty e^{-rt}{}_tp_{x_i}wd(t)E_i\Big[\frac{\pi_i w_i}{\sum_{j} \pi_j w_jN_j(t)}\Big]\,dt.
\end{align*}
We use the notation $E_i$ to remind us that this is a conditional expectation, in which $N_i-1\sim\text{Bin}(n_i-1,{}_tp_{x_i})$, while the other $N_j\sim\text{Bin}(n_j,{}_tp_{x_j})$.
Call the above expression $w_iF_i(\pi_1,\dots,\pi_K)$, so if $\pi=(\pi_1,\dots,\pi_K)$ then
$$
F_i(\pi)=\int_0^\infty e^{-rt}{}_tp_{x_i}wd(t)E_i\Big[\frac{\pi_i }{\sum_{j} \pi_j w_jN_j(t)}\Big]\,dt
$$
represents the present value of the returns per dollar invested by the $i$th cohort. 

Subscribers in the $i$th cohort invest $w_i$, so ideally, ``fairness'' would mean that the present value of each person's payments equals their initial fee, in other words, that $F_i(\pi)=1$ for each $i$. This is not possible, for the simple reason that there is always a positive probability of money being left on the table once everyone dies. Let $A_{i,k}(t)$ be the event that the $k$th individual in the $i$th cohort survives till time $t$. Then $\sum n_iw_i=w$ but
\begin{align*}
\sum n_iw_iF_i(\pi)&=\int_0^\infty e^{-rt}wd(t)\sum_in_i\cdot {}_tp_{x_i}E_i\Big[\frac{\pi_iw_i}{\sum_{j} \pi_j w_jN_j(t)}\Big]\,dt\\
&=\int_0^\infty e^{-rt}wd(t)E\Big[\sum_i\sum_{k=1}^{n_i}\frac{\pi_iw_i}{\sum_{j} \pi_j w_jN_j(t)}1_{A_{i,k}(t)}\Big]\,dt\\
&=\int_0^\infty e^{-rt}wd(t)E\Big[\sum_i\frac{\pi_iw_iN_i(t)}{\sum_{j} \pi_j w_jN_j(t)}1_{\{N_i(t)> 0\}}\Big]\,dt\\
&=\int_0^\infty e^{-rt}wd(t)P\Big(\sum_{j} N_j(t)>0\Big)\,dt
<\int_0^\infty e^{-rt}wd(t)\,dt = w.
\end{align*}
In other words, we have proved that
\begin{lemma}
Regardless of $\pi$, at least one cohort must have its present value $F_i(\pi)<1$. 
\label{lem:Flessthan1}
\end{lemma}
The closest we can come to being truly fair is to have all the $F_i(\pi)$ equal. In other words, each subscriber loses the same tiny percentage of their investment, in present value terms. We say that $\pi$ is {\it equitable} if
$$
F_i(\pi)=F_j(\pi)\, \forall i,j.
$$
Equivalently, 
$$
\text{ $F_i(\pi)=1-\epsilon$ for each $i$, where $\epsilon = \int_0^\infty e^{-rt}d(t)P\Big(\sum_{j} N_j(t)=0\Big)\,dt$}.
$$
If we want to make the tontine fair in the absolute sense, we'd need to return any monies remaining after the last death to the estates of the subscribers (as a whole, or simply to the estate of the last survivor). This is precisely why Donnelly et. al. (2014, 2015) include a death benefit in the products they analyze, which ensures that no money is left-over and allows the designs to be fair. Our approach is to focus exclusively on lifetime income. In other words, we eliminate the death benefit but keep things {\em equitable.}

Since $\sum \frac{n_iw_i}{w}F_i(\pi)=1-\epsilon$ by the above argument, it is clear that either the tontine is equitable, or there are some indices for which $F_i(\pi)>1-\epsilon$ and some for which $F_i(\pi)<1-\epsilon$. We call
$$
\theta(\pi)=max_{i\neq j}|F_i(\pi)-F_j(\pi)|
$$
the {\it inequity} of the tontine. We say that $\pi$ is {\it more equitable} than $\pi'$ if $\theta(\pi)<\theta(\pi')$. 

There is an obstruction to equity, as the following example shows. Suppose $n=K=2$. The first subscriber will receive all the available income during the period they outlive the second subscriber. Therefore if $\frac{w}{w_1}$ is sufficiently large, 
$$
F_1(\pi)> \frac{w}{w_1}\int_0^\infty e^{-rt}d(t){}_tp_{x_1}{}_tq_{x_2}\,dt
> \int_0^\infty e^{-rt}d(t)[1-{}_tq_{x_1}{}_tq_{x_2}]\,dt=1-\epsilon.
$$
For example, using reasonable ages and mortality rates it is impossible to make equitable a tontine in which one subscriber invests one dollar, and another invests a million. The most equitable such a tontine could be is in the limiting case $\pi_1=0$, so that the first subscriber only starts receiving payments once the second subscriber has died. We will address such {\it contingent tontines} in the appendix (Section \ref{sec:appendix}). 

The main theorem of this paper is as follows.

\begin{theorem} Fix $d(t)$ as well as the $n_i$, $x_i$ and $w_i$, $i=1,\dots,K$.  
\begin{enumerate}
\item If there exists an equitable choice of $\pi=(\pi_1,\dots,\pi_K)$ such that $0<\pi_i<\infty$ for each $i$, then this choice is unique up to an arbitrary multiplicative constant. 
\item  A necessary and sufficient condition for such a $\pi$ to exist is the following: 
\begin{equation}
\int_0^\infty e^{-rt}d(t)(\prod_{i\notin A}{}_tq_{x_i}^{n_i})(1-\prod_{i\in A}{}_tq_{x_i}^{n_i})\,dt < \alpha_A(1-\epsilon)
\label{typecondition1}
\end{equation}
for every $A\subset\{1, \dots, K\}$ with $0<|A|<K$, where $\alpha_A=\frac{1}{w}\sum_{k\in A}n_kw_k$. 
\end{enumerate}
\label{thm:existence}
\end{theorem}
We will prove this result in the appendix, where we also expand on the meaning of condition \eqref{typecondition1}. In heuristic terms:  if there is a cohort who find the tontine favourable even if they have to wait for income until all subscribers from other cohorts have died, then equity is impossible.  Note that our formulation tacitly assumed that all members of a cohort share the same participation rate. If equitable rates exist, then this must in fact be the case. To see this, subdivide cohort $i$ into $n_i$ cohorts, each with a single member, and apply the uniqueness conclusion of the above theorem. 

In Section \ref{sec:natural} we examine some plausible scenarios with utility included. Here we treat some extreme examples to illustrate how equitable rates may vary as well as giving some cases in which they fail to exist. We exhibit values in the two-cohort case ($K=2$), using Gompertz hazard rates, i.e. $\lambda_x=\frac{1}{b}e^{\frac{x-m}{b}}$ at age $x$. Parameters are $m=88.72$, $b=10$, and $r=4\%$. In Figure \ref{fig1} we look at age disparities, and in Figure \ref{fig2} we look at disparities in investment levels. 
$$
\text{\bf Insert Figures \ref{fig1} and \ref{fig2} here.}
$$

These figures show the spread in $\pi$ (ratio of the largest to the smallest) narrowing as the population size increases. This is not a general rule however. If in Figure \ref{fig1} we had taken $x_1=90$ and $x_2=65$ the spread would narrow at first but then widen. With $x_1=65$, $x_2=85$, and $n_1=1$ there would be equitable rates for small values of $n_2$ but not for large ones. In Figure \ref{fig2}, the higher the outlier investment $w_2$, the larger the size $n_1$ of the cohort investing $w_1=\$1$ must be, before equitable rates exist. For example, if $w_2=\$20$ we require $n_1\ge5$ for equity to be possible. But $w_2=\$100$ requires $n_1\ge 23$, and $w_2=\$500$ requires $n_1\ge 114$.

Note that equity being infeasible is not purely a phenomenon of small populations. A poorly designed tontine can also produce this effect. For example, suppose we have two cohorts of size $n_1=n_2=100$ with ages $x_1=65$ and $x_2$ to be specified. If each member of the second cohort contributes $w_2=1$, then for a large enough value of $w_1$ the tontine must be inequitable. With a well designed tontine it typically takes a large value of $w_1$ to destroy equity. But if we take a flatter tontine than is desirable -- say a tontine whose $d(t)$ would be natural for a population of age 50 -- then quite modest values of $w_2$ will produce inequity, especially once there is a disparity in cohort ages. For example, if $x_2=80$ (resp. 75/70/65) then even $w_2=7$ (resp. 14/37/209) will accomplish this, according to the Theorem \ref{thm:existence} criterion.\footnote{What this all means practically speaking is that Compton's (1833) scheme to charge different share prices for tontine stock might not work for all ages and investment amounts. The {\em equitable} price is most definitely not linear in the amount invested, which is in contrast to a tontine scheme with homogenous ages. And, while one certainly can't fault Compton (1833) for not realizing this fact, we believe it is an interesting aspect of his rather-clever proposal.}

\section{Utility, asymptotics, and optimality}
\label{sec:utility}

\subsection{Utility and loading factors}
\label{subsec:utility}

For an arbitrary tontine payout $d(t)$ (satisfying \eqref{tontineconstraint} but not necessarily optimal) and arbitrary participation rates $\pi_i$ (not necessarily equitable), we may consider the utility of the cash flow received by an individual from the $i$th cohort.  Namely
$$
\int_0^\infty e^{-rt}{}_tp_{x_i} E_i\Big[U(\frac{wd(t)}{\sum w_j\pi_j N_j(t)}\pi_iw_i)\Big]\,dt.
$$
We are interested in the effect of inhomogeneity in the subscriber population. In particular, we would like to understand whether adding individuals to a tontine raises or lowers utility (and by how much), when the added individuals differ from the rest (in homogeneous populations, adding individuals always increases utility). In particular, for a cohort of size $n_i$ in a heterogeneous tontine with payout $d(t)$, the natural comparison will contrast their utility with that of an optimized tontine $\hat d(t)$ in which only those $n_i$ homogeneous individuals participate. Thus we define a {\it loading factor} $\delta_i$, which (when applied to the homogeneous tontine) makes the two utilities equal. In other words, 
$$
\int_0^\infty e^{-rt}{}_tp_{x_i} E_i\Big[u\Big(\frac{wd(t)}{\sum w_j\pi_j N_j(t)}\pi_iw_i\Big)\Big]\,dt
=\int_0^\infty e^{-rt}{}_tp_{x_i} E_i\Big[u\Big(\frac{n_i\hat d(t)}{ N_i(t)}(1-\delta_i)w_i\Big)\Big]\,dt.
$$
If $\delta_i>0$ this means that the cohort loses utility from the addition of heterogeneous individuals to the pool. If $\delta_i<0$ then the cohort gains utility from the addition of these individuals. For a different comparison, between tontines and annuities, see Milevsky and Salisbury (2015), where a different loading factor is used. See also the related work by Hanewald et. al. (2015) which examines how product loadings might affect the choice between different mortality-contingent claims.

In Section \ref{sec:natural}, we will give numerical calculations of loadings for various choices of $d(t)$, and we will see that in a well-designed tontine, adding participants increases utility (i.e. loadings are negative). We will work with $\gamma=1$ so $U(c)=\log c$, and the above formula simplifies considerably. By results of Milevsky and Salisbury (2015), the optimal $\hat d(t)$ is the tontine that is natural for the age-$x_i$ cohort, in other words, $\hat d(t)=\frac{1}{a_{x_i}}{}_tp_{x_i}$. 

\subsection{Asymptotics and the proportional tontine}
\label{subsec:proportional}

We fix $K$, the $x_i$, and the $w_i$ and consider the limit of the $\pi_i$ when the total number of subscribers $n=\sum n_i\to\infty$. Let $\alpha_i>0$ and $\sum_{i=1}^K\alpha_i=1$. Assume that the $n_i\to\infty$ in such a way that $\frac{n_iw_i}{w}\to\alpha_i$, so $\alpha_i$ represents the fraction of the initial investment attributable to the  $i$th cohort. Then
$$
F_i(\pi)=
\int_0^\infty e^{-rt}{}_tp_{x_i}wd(t)E_i\Big[\frac{\pi_i }{\sum \pi_j w_jN_j(t)}\Big]\,dt.
\to
\int_0^\infty e^{-rt}d(t)\frac{\pi_i \cdot{}_tp_{x_i}}{\sum_{j=1}^K \pi_j \alpha_j\cdot{}_tp_{x_j}}\,dt.
$$
A particular case of this requires particular attention. Let $a_x=\int_0^\infty e^{-rt}{}_tp_{x}\,dt$ be the standard annuity price of \$1 for life for age $x$ individuals. Consider $d(t)=\sum_j \frac{\alpha_j}{a_{x_j}}{}_tp_{x_j}$ (which clearly satisfies the condition that $\int_0^\infty e^{-rt}d(t)=1$).  In this case, $F_i(\pi)\to a_{x_i}\pi_i$, so the equitable participation rates asymptotically become $\pi_i=\frac{1}{a_{x_i}}$. We call a tontine with $d(t)=\sum\frac{n_jw_j}{w}\times\frac{{}_tp_{x_j}}{a_{x_j}}$ and  $\pi_i=\frac{1}{a_{x_i}}$ a {\it proportional tontine}, and emphasize that it is equitable only in the limit as $n\to\infty$. In the case of a homogeneous population (i.e. $K=1$), the proportional tontine agrees with what we have earlier called the natural tontine for this cohort. 

One motivation for this particular design is that the payout rate to a surviving individual from the $i$th group, at time $t$, is asymptotically
$$
\frac{d(t)}{\sum \pi_j\alpha_j\cdot{}_tp_{x_j}}\pi_i=\pi_i=\frac{1}{a_{x_i}}
$$
per unit. In other words, the rate of payment to a surviving individual remains constant in time, and is simply the standard annuity factor of $\frac{1}{a_{x_i}}$ per dollar of initial premium. In this sense, a proportional tontine reproduces (in the limit) the payment structure and cost of a standard fixed annuity for each subscriber. We will shortly see a further motivation, when we show that it is asymptotically optimal. In Section \ref{sec:others} we will connect this design to the group self-annuitization scheme (GSA). 

How do our utility loadings behave when $n\to\infty$ as above? The above equation becomes that
$$
\int_0^\infty e^{-rt}{}_tp_{x_i} u\Big(\frac{d(t)}{\sum \pi_j\alpha_j\cdot{}_tp_{x_j}}\pi_iw_i\Big)\,dt
=\int_0^\infty e^{-rt}{}_tp_{x_i} u\Big(\frac{\hat d(t)}{ {}_tp_{x_i}}(1-\delta_i)w_i\Big)\,dt.
$$
Take $d(t)$ to be the proportional tontine, so in the limit, so $\pi_j=\frac{1}{a_{x_j}}$ is equitable in the limit. As above, take $u$ to be logarithmic, and $\hat d(t)$ to be natural for the age-$x_i$ cohort. We obtain that
$$
\int_0^\infty e^{-rt}{}_tp_{x_i} u\Big(\frac{w_i}{a_{x_i}}\Big)\,dt
=\int_0^\infty e^{-rt}{}_tp_{x_i} u\Big(\frac{w_i}{a_{x_i}}(1-\delta_i)\Big)\,dt,
$$
from which we immediately get the following:
\begin{lemma}
Asymptotically, the proportional tontine has utility loadings $\delta_i=0$. 
\label{lem:zeroloading}
\end{lemma}

\subsection{Can a tontine be optimal for multiple cohorts?}
\label{subsec:optimality}

A natural question is whether it is possible to design a tontine to be optimal for multiple age cohorts. This turns out not to be possible, except in the limit as $n\to\infty$. To formulate the question, we include equity as an additional set of constraints in the optimization problem. In particular, we wish to choose $d(t)$ and the $\pi_j$ to maximize the utility of the $i$th cohort
$$
\int_0^\infty e^{-rt}{}_tp_{x_i} E_i\Big[U(\frac{wd(t)}{\sum w_j\pi_j N_j(t)}\pi_iw_i)\Big]\,dt
$$
over $d(t)\ge 0$, subject to the budget constraint $\int_0^\infty e^{-rt}d(t)\,dt=1$ and the equity constraints
$$
\int_0^\infty e^{-rt}{}_tp_{x_i}wd(t)E_i\Big[\frac{\pi_i }{\sum_{j=1}^K \pi_j w_jN_j(t)}\Big]\,dt
=\int_0^\infty e^{-rt}{}_tp_{x_\ell}nd(t)E_\ell\Big[\frac{\pi_\ell}{\sum_{j=1}^kK\pi_j w_jN_j(t)}\Big]\,dt
$$
for $\ell\neq i$. 

In the limit as $n\to\infty$ we wish to maximize 
$$
\int_0^\infty e^{-rt}{}_tp_{x_i} U(\frac{d(t)}{\sum \alpha_j\pi_j \cdot{}_tp_{x_j}}\pi_iw_i)\,dt
$$
over $d(t)\ge 0$, subject to the budget constraint $\int_0^\infty e^{-rt}d(t)\,dt=1$ and the equity constraints
$$
\int_0^\infty e^{-rt}{}_tp_{x_i}d(t)\frac{\pi_i }{\sum_{j=1}^K \alpha_j\pi_j \cdot{}_tp_{x_j}}\,dt
=\int_0^\infty e^{-rt}{}_tp_{x_\ell}d(t)\frac{\pi_\ell }{\sum_{j=1}^K \alpha_j\pi_j \cdot{}_tp_{x_j}}\,dt
$$
for $\ell\neq i$. This version of the problem simplifies if reformulated in terms of $\Gamma(t)=\frac{d(t)}{\sum_{j=1}^k \alpha_j\pi_j \cdot{}_tp_{x_j}}$. Now we seek to maximize
$\int_0^\infty e^{-rt}{}_tp_{x_i} U(\pi_iw_i\Gamma(t))\,dt$
over $\Gamma(t)\ge 0$, subject to the budget constraint $\int_0^\infty e^{-rt}\Gamma(t)\sum \alpha_j\pi_j \,{}_tp_{x_j}\,dt=1$ and the equity constraints
$\int_0^\infty e^{-rt}\Gamma(t)\pi_i\,{}_tp_{x_i}\,dt
=\int_0^\infty e^{-rt}\Gamma(t)\pi_\ell\,{}_tp_{x_\ell}\,dt$
for $\ell\neq i$. 

The equity constraints become merely that 
$$
\pi_\ell = \pi_i\frac{\int_0^\infty e^{-rt}\Gamma(t)\,{}_tp_{x_i}\,dt}{\int_0^\infty e^{-rt}\Gamma(t)\,{}_tp_{x_\ell}\,dt}.
$$
and substituting back, the budget constraint becomes that $\pi_i\int_0^\infty e^{-rt}\Gamma(t)\,{}_tp_{x_i}=1$. This puts us back in the context of optimizing the simple annuity of Theorem \ref{thm:annuity}, which implies that the optimal $\Gamma(t)$ is constant. If we normalize so $\pi_i=\frac{1}{a_{x_i}}$ then the $\pi_\ell=\frac{1}{a_{x_\ell}}$, and we get $\Gamma(t)=1$. 

In particular, optimizing the utility of the $i$th cohort, in the presence of equity constraints, asymptotically gives precisely the proportional tontine described in the last section, i.e. $d(t)=\sum_j \frac{\alpha_j}{a_{x_j}}{}_tp_{x_j}$. Therefore this optimal tontine (in this case, really a type of annuity) has the same design, regardless of which $i$ one chooses to optimize for. We have shown that
\begin{proposition} Assume a strictly concave utility function. 
In the limit as $n\to\infty$, the proportional tontine optimizes the utility of each cohort simultaneously.
\label{prop:asymptoticoptimality}
\end{proposition}

The original optimization problem (i.e. in the setting of finite $n$) can also be solved, though not so cleanly. We do not present this here, except to note that when we optimize even the logarithmic utility of the $i$th cohort, the results turn out to no longer be consistent when we vary $i$. In other words, it is typically impossible to make everyone happy simultaneously. This is one reason we feel it is reasonable to first fix a tontine structure $d(t)$ (as we have done above), and then allow people participate at equitable rates if they so wish. Naturally, this means one of the questions we will need to answer is how significant their utility loss is, when doing so. 

\section{Our Suggested $d(t)$: The Natural and Equitable Tontine}
\label{sec:natural}
In the context of a homogeneous population of age $x$, all investing equal amounts, the design proposed in Milevsky and Salisbury (2015) had $d(t)=\frac{1}{a_x}{}_tp_x$, i.e. the natural tontine for age $x$. In this context, the design is optimal in the case of logarithmic utility, and near-optimal otherwise. In this section, we wish to propose a suitable generalization in the heterogeneous setting. 

For heterogeneous tontines, we have seen that overall optimality is not feasible (except asymptotically). In that context, we propose adopting the following design, which performs well in numerical experiments we have conducted, reduces to the above design in the case of a homogeneous population, and agrees with the proportional tontine in the limit as $n\to\infty$ (so is optimal asymptotically).

Fix the $x_i$, $w_i$, and $n_i$. We say that a tontine is {\it natural} if $d(t)$ is at all times proportional to the mean number of surviving tontine shares. In other words, $d(t)=c\sum u_j n_j\cdot {}_tp_{x_j}=c\sum \pi_jw_jn_j\cdot {}_tp_{x_j}$. Integrating, we see that 
$$
d(t)
=\sum_i\Big[\frac{\pi_in_iw_i}{\sum_ja_{x_j}\pi_jw_jn_j}\Big]{}_tp_{x_i}.
$$
Note that once the $\pi_i$ are given, the natural tontine is fully determined by the budget constraint. But to construct a tontine that is both natural and equitable, we must compute the $\pi_i$ and $d(t)$ simultaneously. In practice this is more complicated than (as above) simply fixing a $d(t)$ and computing equitable $\pi$'s, but not unduly so (at least when the number $K$ of types is small). 

The following two tables (Table \ref{table03} for $K=2$ cohorts, Table \ref{table04} for $K=3$) 
$$
\text{\bf Insert Tables \ref{table03} and \ref{table04} here.}
$$
display such natural and equitable tontines, and compare them to ``natural'' tontines that would have been chosen if the population had been homogeneous (but with equitable participation rates). We also compare with the corresponding proportional tontines, though those are not equitable. Though the theoretical basis of proportional tontines is not as appealing as that of natural ones, they are simpler to compute, and they do appear to perform reasonably in practice. We view them as an acceptable alternative if computational resources are not available to work out equitable $\pi_i$'s and natural $d(t)$'s.  Note that since these tables normalize $\pi$ to make $\pi_i=1$ for some $i$, this means that the proportional tontine has $\pi_j=a_{x_i}/a_{x_j}$.

First consider Table \ref{table03}. Rows labelled ``A'' and ``D'' use tontine designs that would be natural for homogeneous populations, of age 65 and 75 respectively. Equitable $\pi$'s are then computed. In row A, both $\delta_1$ and $\delta_2$ start negative ($n_1=1=n_2$), meaning that the benefit of the extra participant outweighs the impact of heterogeneity. As the common value of $n_1=n_2$ rises, the $\delta_i$ become positive and defects in the design become more relevant. In particular, loadings remain strictly positive asymptotically -- while the product becomes essentially an annuity, it is not an optimal one. Note that the loadings are not actually monotone. Adding participants is more beneficial for older (age 75) participants than for younger (age 65) ones. Surprisingly, in the presence of age 65 participants, even the age 75 ones get more benefit from an age-65 design over an age-75 one. 

Rows labelled ``B'' correspond to a truly natural and equitable design. Rows labelled ``C'' are proportional designs. In most cases, either performs better for both cohorts than the homogeneous designs do. The problem with the A design is now clear -- to be equitable it requires a higher participation rate $\pi_2$, which dilutes the benefit of adding individuals to the tontine, and produces utility loss. In contrast, rows B and C typically show a negative loading (i.e. a utility gain), though with different choices of parameters (not shown) this can in fact sometimes not be the case. 

Comparing B and C, it is generally the case that the older cohort prefers a natural design, whereas the younger cohort prefers a proportional design. The proportional design comes closer to equalizing the utility gains between the cohorts. The two designs perform similarly asymptotically. The factor contributing most to their difference is the equitability of $\pi$ rather than the choice of $d(t)$ -- using the proportional $d(t)$ but equitable $\pi$'s would turn out to give very similar utilities to the fully natural design. 

Table \ref{table04} treats the three-cohort case, comparing the natural and proportional designs with a design that would be natural for the age-65 cohort alone. Now all three designs have very similar effects on utility. Otherwise the table is consistent with empirical observations made above: adding people to the tontine is generally favourable (despite heterogeneity); and the utility improvement is greater for the older participants. Note that the good performance of these designs may in part be a consequence of a balance between ages 60 and 70 -- asymetric designs (not shown) are less consistent.

Figure \ref{fig03} shows a simulation of the payouts from a 2-cohort natural and equitable design with $n_1=200$ members and $n_2=50$ members. Note that at moderate ages it comes close to achieving a constant and steady payout to each survivor. There is higher volatility in payments at advanced ages, once the number of survivors in the pool is small. Mitigating that volatility would be a requirement for a practical tontine design. In fact, we believe this is not hard to achieve, for the following reason: our tontines are designed to be optimal when individuals have no exogenous income. In reality there is typically some exogenous pension income (eg. Social Security). Unpublished work by Ashraf (2015) suggests that in the presence of exogenous income, optimal tontines should be designed to taper off and cease payments at advanced ages (eg by age 100). If this is done, then by the time the survivor pool is very small, any variability in its size will no longer matter. 
$$
\text{\bf Insert Figure \ref{fig03} here.}
$$

Natural and equitable tontines appear to exist for a broad range of parameter values (though not universally -- the obstruction raised in Section \ref{sec:mixing} still remains valid). But we have not yet succeeded in finding necessary and sufficient conditions for existence, or in establishing uniqueness, as in Theorem \ref{thm:existence}. Therefore resolving the following remains a topic for further research. 

\begin{conjecture} 
Fix the $x_i$, $n_i$, and $w_i$, $i=1,\dots,K$. Under broad conditions there will exist a choice of $\pi=(\pi_1,\dots,\pi_K)$ such that the corresponding natural tontine is equitable. Up to an arbitrary multiplicative constant, there is at most one such $\pi$. 
\end{conjecture}

\section{Other product designs}
\label{sec:others}
As indicated earlier, there are a number of other product designs in the actuarial literature that hedge the idiosyncratic component of longevity risk but not the systematic component. In this section we discuss some of those alternatives.
\subsection{Pooled Annuity Fund (PAF)}

In the homogeneous setting, the optimal PAF was derived in Stamos (2008), and its utility (or loading in our terminology) compared to a (variable) life annuity was investigated in Donnelly et al (2013). We are not aware of work on such optimal PAF's in the heterogeneous setting, though an approach like that of this paper (i.e. fix a payout mechanism and then allocate shares equitably) could probably be carried out in this context. PAF's in general allow a diversified investment portfolio, but we will consider only risk-free portfolios. In other words, this section treats PAF's invested purely in bonds (at rate $r$) with a homogeneous pool of subscribers. 

A PAF allows the rate $e(t,k,w)$ at which each individual is paid to vary with $t$, but also with the number of survivors $k=N_t$ and with the individual's share $w=W_t$ of total assets under management $\overline{w}=\overline{W}_t$ (so $\overline{W}_t=N_tW_t$ and $\overline{w}=kw$). 
For a given risk-aversion coefficient $\gamma\neq 1$, Stamos (2008) obtains the utility-optimizing payout rates $e$, and shows that they take the form $e(t,k,w)=\eta(t,k)w$ for some function $\eta(t,k)$. The extra flexibility means this provides higher utility than a tontine (where dependence on $k$ and $w$ is not allowed, other than indirectly via the initial number $n$ of subscribers). Table \ref{table09} indeed shows a modest improvement.

It comes at the expense of dealing with a more complex product. For example, a prospectus would have to provide the full table of $\eta(t,k)$'s and that complexity may make it harder for subscribers to understand (and then manage) the risks associated with the product. The payout to an individual will no longer be predictable in terms of the current number $N_t$ of survivors, i.e. it will be path-dependent. A general PAF will also require the fund manager to adjust the portfolio in response to the observed mortality experience of the pool, rather than relying on a static bond portfolio. So there are both advantages and disadvantages to each design. 

To compare with our tontines' payout $d(t)$, we let $\overline{e}(t,k,\overline{w})$ denote the total payout rate, per initial dollar invested. We assume that each individual invests \$1 initially, so $\overline{e}(t,k,\overline{w})=\eta(t,k)\overline{w}/n$.

An observation, that we have not seen recorded in the literature, is that when $\gamma=1$ (in the setting described above), the total withdrawals from the optimal PAF become deterministic. In other words, this PAF is a tontine (and therefore must be the natural tontine):
\begin{proposition} Assume a homogeneous pool, with assets invested risk-free at rate $r$. Assume logarithmic utility. Then the optimal PAF has $\overline{e}(t,N_t,\overline W_t)=\frac{1}{a_x}{}_tp_x$.
\label{prop:PAFandtontine}
\end{proposition}
We give the proof in the appendix. To give a sense of the differences, Figure \ref{fig03} simulates both total payouts $d(t)$ and $\overline{e}(t,N_t,\overline{W}_t)$, and individual payouts $d(t)/N_t$ and $e(t,N_t,W_t)$, for $\gamma=5$, $x=65$, and a small pool of size $n=10$. The utility improvement from the PAF appears to derive from a modest reduction in the volatility of individual payouts. With a larger pool of size $n=100$, the difference in either total or individual payouts (not shown) become negligible, except at quite advanced ages.
$$
\text{\bf Insert Figures \ref{fig03} and \ref{fig04} here.}
$$

\subsection{Group Self Annuitization (GSA)}

The GSA scheme was proposed in Piggott, Valdez, and Detzel (2005). It provides a rule for managing payments from a pool of assets, so does not depend on risk aversion, nor does it attempt to optimize utility. On the other hand, it allows for a heterogeneous pool, and variable asset returns, though as in other sections of this paper, we will focus here on the case of a fixed rate of return $r$. 

The general GSA scheme also allows new individuals to join over at times $t>0$, by valuing each survivor's share of current assets, and then allowing new individuals to buy in at an actuarially fair price.  A similar approach could be used to add this feature to the tontines considered in the current paper, but we do not pursue this idea here. We will therefore treat only GSAs in which all investors buy-in at time 0. 

A GSA scheme works as follows, in discrete time: members of the $i$th cohort each contribute $w_i$ at time 0. This entitles them to an initial payment  matching that an annuity would provide. At later times $t_k$, everyone's payment is adjusted up or down by a common annuity factor $M_k$ chosen so that if realized mortality were to match expected mortality thereafter, no further adjustment would be required. In symbols, survivors from the $i$th cohort receive $g_{i,k}$ at time $t_k=k\,\Delta t$. Initially $g_{i,0}=w_i/\dot{a}_{x_i}$ where $\dot{a}_x$ is the discrete annuity price
for \$1 each period $\Delta t$, i.e. $\dot{a}_x=\sum_{k=0}^\infty e^{-rt_k}{}_{t_k}p_x$.  Later, $g_{i,k}=M_kg_{i,0}$ where $\sum_i g_{i,k}N_i(t_k)\dot{a}_{x_i+t_k}=W_k$ and $W_k$ is the wealth at time $t_k$. The latter is determined recursively as $W_{k+1}=(W_k-g_k)e^{r\Delta t}$ where $g_k=\sum_ig_{i,k}N_i(t_k)$ is the total disbursed at time $t_k$. 

Suppose first that all cohorts have the same age $x_1$, but possibly different initial contributions $w_i$. Then $g_k=\frac{W_k}{\dot{a}_{x_1+t_k}}$ so the above recurrence implies that $g_k$ and $W_k$ are both deterministic (i.e. tontine-like). We have not seen this conclusion recorded in the literature. Moreover 
$g_{k+1}=\frac{W_{k+1}}{\dot{a}_{x_1+t_{k+1}}}=\frac{W_k-g_k}{\dot{a}_{x_1+t_{k+1}}}e^{r\Delta t}=g_k\frac{\dot{a}_{x_1+t_k}-1}{\dot{a}_{x_1+t_{k+1}}}e^{r\Delta t}$. From a standard actuarial recursion, this implies that $g_{k+1}=g_k \,{}_{\Delta t}p_{x_1+t_k}$, from which we obtain
\begin{proposition}
Assume a pool whose cohorts have the same initial age $x_1$ and contribute $w_i$ at time 0. Assume that assets are invested risk-free at rate $r$. Then the total GSA payout is $g_k=\frac{\sum n_jw_j}{\dot{a}_{x_1}}\,{}_{t_k}p_{x_1}$, while individual payouts are $g_{i,k}= \frac{g_kw_i}{\sum_j w_jN_j(t_k)}$. 
\end{proposition}
If we now take limits as $\Delta t\downarrow 0$, we obtain a natural tontine (for age $x_1$) that pays out continuously at rate $\frac{\sum n_jw_j}{a_{x_1}} {}_tp_{x_1}$. In the language we introduced earlier, the participation rates for this tontine are all $\pi_i=1$, so it is not equitable (except in the homogeneous case of a single cohort). This is consistent with the results of Donnelly (2015) who shows that a GSA scheme is only fair in the homogeneous case. As noted earlier, to achieve fairness, she has to include payments to estates. In this case, during the final period when anyone is alive, all who die are deemed to receive the remaining assets distributed in the proportions set by the GSA rules.

In the fully heterogeneous case (variability in both initial ages and initial investments), there is no reason that either wealth or total withdrawals should be deterministic. In other words, the GSA is no longer tontine-like. It turns out that the appropriate tontine to compare with in this case is what we have called a proportional tontine. To start with, the GSA payouts (resp. proportional tontine payout rates) are always proportional to the initial payout $g_{i,0}=\frac{w_i}{\dot{a}_{x_i}}$ (resp. initial payout rate $\frac{w_i}{{a}_{x_i}}$). In fact, any deviation between the GSA and proportional tontine payouts derives either from the discrete-time formulation of the GSA, or from deviations of the survival counts $N_i(t)$ from their means -- it can be shown (though we will not give details) that if each $N_i(t_k)$ agreed with $n_i\times{}_{t_k}p_{x_i}$, then in the limit as $\Delta t\downarrow 0$, the two cash flow streams would agree precisely. More concretely, as Figure \ref{fig05} shows, the actual payment streams are quite close even for small $n$.
$$
\text{\bf Insert Figure \ref{fig05} here.}
$$
\subsection{Other designs}
Donnelly, Guillen and Nielsen (2014) formulate another design, known as an {\it annuity overlay fund } (AOF), as a way of pooling individual investment accounts in order to capture mortality credits. The overlay pays out the assets of individuals who die in a period, in proportion to all who belong to the pool at the beginning of the period (including those who die). It does so in proportion to both the individual's assets at risk in the pool, and to the individual's hazard rate. 

The AOF is designed to work with arbitrary investment and withdrawal decisions by participants, so it seeks to achieve actuarial fairness over every period, as opposed to simply over the lifecycle. This is a very different objective than that of a tontine or annuity, whose goal is providing stable lifetime income. We should therefore not expect the two designs to behave similarly. 

Another design is the {\it fair transfer plan} (FTP) of Sabin (2010). Only living participants receive payments. So (as in the current paper) Sabin's goal is not actuarial fairness, but rather to ensure that no individual has an advantage over another (i.e. what we call equitability). On the other hand, he requires this to be achieved over every period (as in Donnelly et al (2014)), rather than once over the lifecycle, which means that his design is not comparable to ours. We note that he does obtain necessary and sufficient conditions for the existence of an equitable FTP, in his context. 

\section{Conclusion}
\label{sec:conclusion}

There is a growing interest among practitioners and academics in the optimal design of retirement de-accumulation products that insure against idiosyncratic longevity risk while sharing aggregate exposure within a group. In this paper we investigated the design of a retirement income tontine scheme that allows individuals with different mortality rates to participate in the same pool. And while this scheme might not be actuarially {\em fair}, in the sense of Donnelly (2015), this scheme is {\em equitable} in that the scheme does not discriminate against any particular sub-group and all participants receive the exact same expected present value of benefits. And although alternative designs can sometimes provide somewhat greater utility than a tontine, the retirement income tontine has the advantage of transparency, simplicity and requiring little if any actuarial expertise to operate. It pays a reasonably steady and predictable cash flow to a declining group of survivors. It is also simpler to analyze qualitatively, and leads to interesting mathematical properties and insights. 

The structure we introduce in this paper -- which is an extension of Milevsky and Salisbury (2015) -- allows anyone of any age to participate in the scheme by adjusting the price of a tontine share to be a function of (i.) the number of investors, (ii.) their ages, and (iii.) the capital they have invested. In Lorenzo Tonti's original scheme, as well as the structure proposed in Milevsky and Salisbury (2015), all investors (in the same pooling class) were assumed to be of the same age and paid the same price. When smaller groups were segmented into age bands they lost the benefit of large numbers. In this paper we have proved that it is possible to mix cohorts without discriminating provided the diversity of the pool satisfies certain dispersion conditions and we propose a specific design that appears to work well in practice.

Finally, this paper provides a detailed comparison of the various mortality pooling schemes that have been proposed in the literature as well as the conditions under which they all collapse into a tontine-like structure. Indeed, regardless of what they are called in practice they all do seem to share a common ancestor.

\newpage

%\begin{figure}[here]
%\begin{center}
%\includegraphics[width=0.7\textwidth]{figure1.jpg} 
%\caption{The payout function $d(t)$ can be visualized as a time-dependent yield which declines in proportion to the assumed survival probability for the group at time zero.}
%\label{fig:singlecohort}
%\end{center}
%\end{figure}

\begin{table}
\begin{center}
\begin{tabular}{||c||c|c||c|c||}
\hline\hline
\multicolumn{5}{||c||}{\textbf{Certainty equivalents for \$100}} \\ \hline\hline
$n$ & 10 & 100 & 10 & 100\\ \hline\hline
& \multicolumn{2}{|c||}{$\gamma =0.5$} & \multicolumn{2}{|c||}{$\gamma =1$}  \\ \hline
A & 101.53 & 100.15 & 102.68 & 100.28\\ \hline
B & 101.55 & 100.15 & 102.68 & 100.28\\ \hline
C & 101.67 & 100.15 & 102.68 & 100.28\\ \hline\hline
& \multicolumn{2}{|c||}{$\gamma =2$} & \multicolumn{2}{|c||}{$\gamma =5$}  \\ \hline
A & 104.62 & 100.53 & 109.20 & 101.22 \\ \hline
B & 104.65 & 100.53 & 109.47 & 101.24 \\ \hline\hline
\multicolumn{5}{||c||}{\footnotesize Assumes $r=4\%$ and Gompertz Mortality} \\ 
\multicolumn{5}{||c||}{\footnotesize   ($m=88.72,b=10$). Homogeneous pool }\\
\multicolumn{5}{||c||}{\footnotesize of size $n$, with initial age 65. } \\ 
\hline\hline
\end{tabular}\medskip
\caption{Shows the amount invested in three products needed to yield the same utility as \$100 invested in a life annuity guaranteed by an insurance company. Product designs are:
\newline A = Pooled Annuity fund, $\gamma$-optimized; Donnelly, Guillen and Nielsen (2013)
\newline B = Tontine, $\gamma$-optimized; Milevsky and Salisbury (2015)
\newline C = Group Self Annuitization Scheme; Piggott, Valdez and Detzel (2005) }
\label{table09}
\end{center}
\end{table}

\hbox{ }

\newpage

\begin{figure}[h]
\begin{center}
\includegraphics[width=0.85\textwidth]{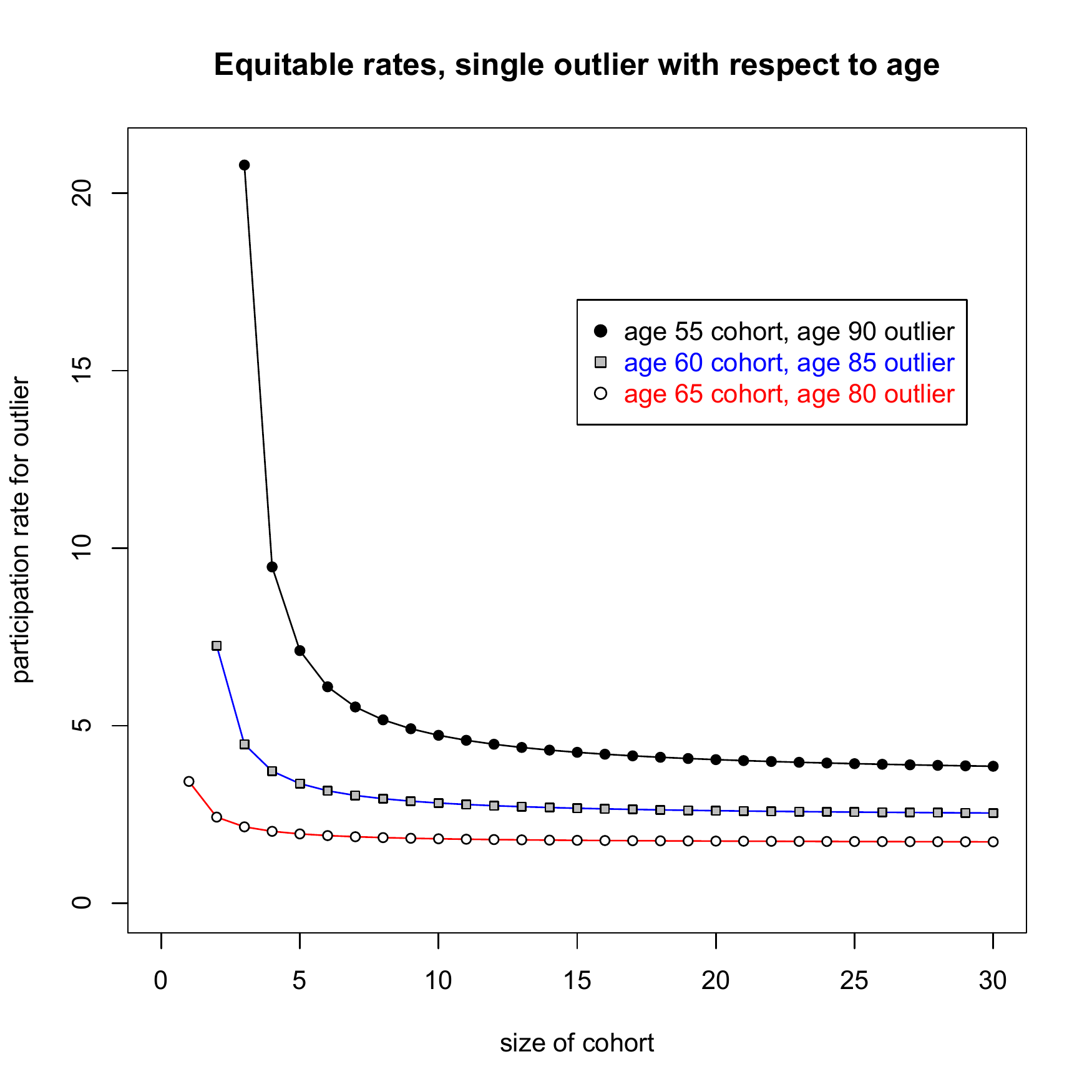} 
\caption{Shows the equitable participation rate $\pi_2$ versus the size $n_1$ of the cohort with age $x_1$, in the presence of a single outlier ($n_2=1$) with age $x_2$. Tontine is natural for the age $x_1$ cohort, and each individual invests \$1 ($w_1=w_2=1$). Normalized so $\pi_1=1$. Assumes Gompertz Mortality ($m=88.72,b=10$) and $r=4\%$.}
\label{fig1}
\end{center}
\end{figure}

\hbox{ }

\newpage

\begin{figure}[h]
\begin{center}
\includegraphics[width=0.85\textwidth]{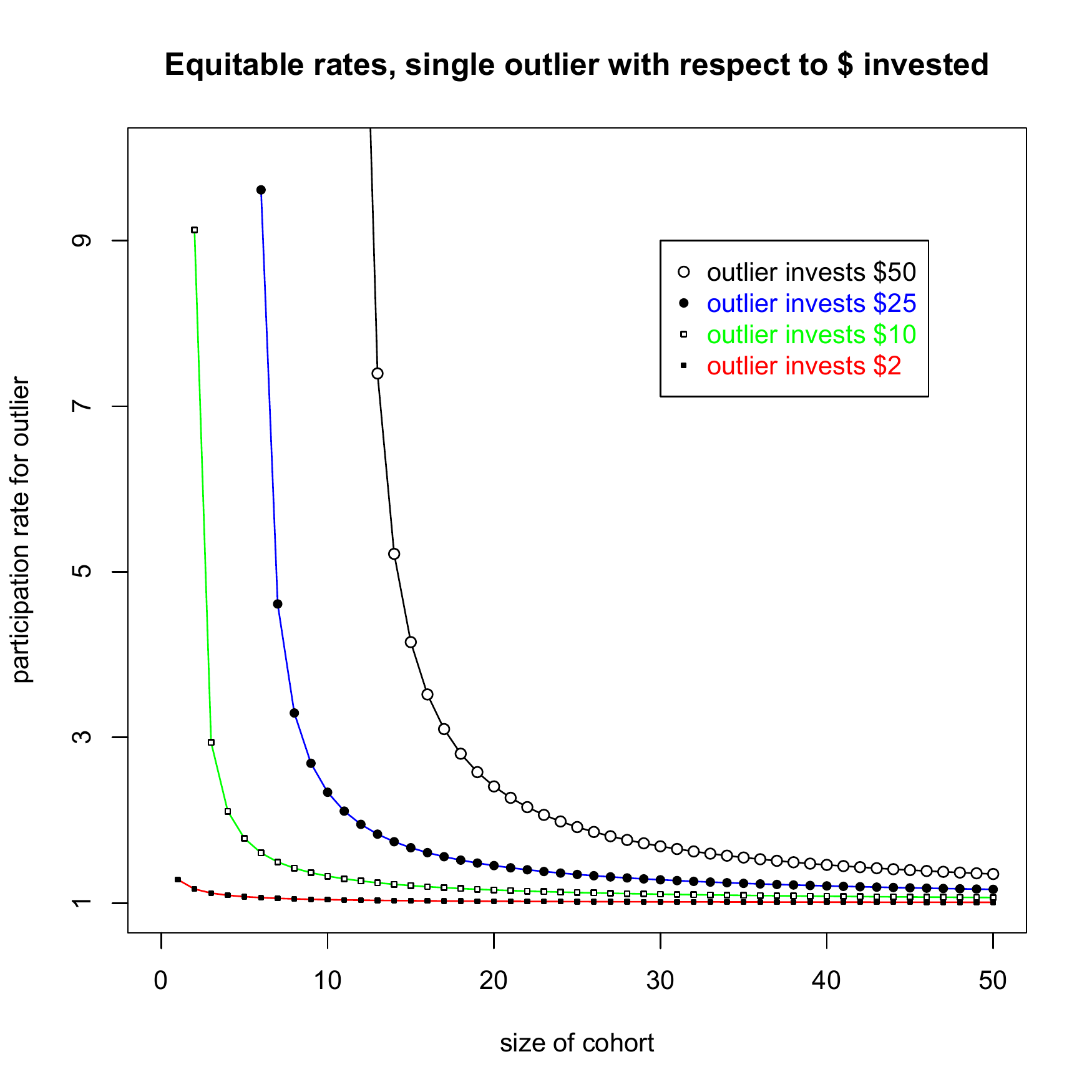} 
\caption{Shows the equitable participation rate $\pi_2$ versus the size $n_1$ of the cohort who invest 1 dollar ($w_1=1$) each, in the presence of a single outlier who invests $w_2$ dollars. All subscribers are the same age ($x_1=x_2=65$), and the tontine is natural for that age. Normalized so $\pi_1=1$. Assumes Gompertz Mortality ($m=88.72,b=10$) and $r=4\%$.}
\label{fig2}
\end{center}
\end{figure}

\hbox{ }

\begin{table}
\begin{center}
\begin{tabular}{||c||c|c|c||c|c|c||}
\hline\hline
\multicolumn{7}{||c||}{\textbf{Equitable rates and loadings in pools  }} \\ 
\multicolumn{7}{||c||}{\textbf{ with $K=2$ cohorts: $n_1=n_2$ }} \\ \hline\hline
  & age 65 & \multicolumn{2}{|c||}{ age 75} & age 65 & \multicolumn{2}{|c||}{ age 75} \\ \hline\hline
  & $\delta_1$ & $\delta_2$ & $\pi_2$ & $\delta_1$ & $\delta_2$ & $\pi_2$ \\ \hline\hline
  & \multicolumn{3}{|c||}{ $n_1=1=n_2$} & \multicolumn{3}{|c||}{ $n_1=50=n_2$} \\ \hline
A & -235.4 & -2604.4 & 1.829 & 239.4 & 30.0 & 1.501  \\ \hline
B & -495.0 & -2819.3 & 1.631 & -3.7 & -69.8 & 1.375 \\ \hline
C & -1266.7 & -2012.0 & 1.370 & -20.6 & -52.9 & 1.370 \\ \hline
D & 277.7 & -2759.3 & 1.506 & 696.1 & 74.3 & 1.265 \\ \hline\hline
  & \multicolumn{3}{|c||}{ $n_1=5=n_2$} & \multicolumn{3}{|c||}{ $n_1=500=n_2$}\\ \hline
A & 177.7 & -496.8 & 1.550 & 240.0 & 92.8 & 1.495 \\ \hline
B & -69.7 & -612.3 & 1.413 & -0.22 & -7.7 & 1.371 \\ \hline
C & -219.9 & -458.7 & 1.370 & -2.0 & -5.9 & 1.370 \\ \hline
D & 646.5 & -485.6 & 1.302 & 700.2 & 135.7 & 1.262 \\ \hline\hline
  & \multicolumn{3}{|c||}{ $n_1=10=n_2$} & \multicolumn{3}{|c||}{ $n_1=n_2\to\infty$}\\ \hline
A & 218.4 & -213.3 & 1.523 & 239.7 & 100.7 & 1.494 \\ \hline
B & -28.9 & -317.9 & 1.392 & 0 & 0 &  1.370 \\ \hline
C & -106.3 & -239.5 & 1.370 & 0 & 0 & 1.370 \\ \hline
D & 676.4 & -179.5 & 1.281 & 700.7 & 143.2 & 1.261 \\ \hline\hline
\multicolumn{7}{||c||}{\footnotesize Assumes $r=4\%$, Gompertz Mortality ($m=88.72,b=10$);} \\ 
\multicolumn{7}{||c||}{\footnotesize $\delta_i$ are given in b.p.; Rates are normalized so $\pi_2=1$} \\ \hline\hline
\end{tabular}\medskip
\caption{Shows the participation rates $\pi_2$ (= inverse of the share price) and corresponding utility loadings $\delta_1,\delta_2$ when there are two cohorts of subscribers in the pool: ages $x_1=65$ and $x_2=75$. Utility is logarithmic. Everyone invests 1 dollar ($w_1=w_2=1$). Tontine designs are: 
\newline A = Natural tontine based on age 65 cohort alone, equitable rates;
\newline B = Natural tontine based on the range of ages, equitable rates;
\newline C = Proportional tontine;
\newline D = Natural tontine based on age 75 cohort alone, equitable rates;}
\label{table03}
\end{center}
\end{table}

\begin{table}
\begin{center}
\begin{tabular}{||c||c|c|c||c|c|c||}
\hline\hline
\multicolumn{7}{||c||}{\textbf{Participation rates and utility loadings in pools}} \\
\multicolumn{7}{||c||}{\textbf{with three cohorts: $2n_1=n_2=2n_3$}} \\ \hline\hline
& age 60 & age 65 & age 70 & age 60 & age 65 & age 70  \\ \hline\hline
& $\pi_1$ & $\pi_2$ & $\pi_3$ & $\delta_1$ & $\delta_2$ & $\delta_3$  \\ \hline\hline
& $n_1=5$  & $n_2=10$  & $n_3=5$  & $n_1=5$  & $n_2=10$  & $n_3=5$   \\ \hline
A & 0.886  & 1 & 1.161 & -186.9 & -136.1 & -594.3 \\ \hline
B & 0.884  & 1 &  1.161 & -216.0 & -136.6 & -586.8\\ \hline
C & 0.889 & 1 & 1.153 & -275.0 & -138.7 & -586.8 \\ \hline\hline
& $n_1=10$  & $n_2=20$  & $n_3=10$  & $n_1=10$  & $n_2=20$  & $n_3=10$   \\ \hline
A & 0.889 & 1 & 1.157 & -79.4 & -68.9 & -301.0 \\ \hline
B & 0.887 & 1 & 1.157 & -102.9 & -70.4 & -297.2 \\ \hline
C & 0.889 & 1 & 1.153 & -133.3 & -71.3 & -264.5 \\ \hline\hline
& $n_1=20$  & $n_2=40$  & $n_3=20$  & $n_1=20$  & $n_2=40$  & $n_3=20$   \\ \hline
A & 0.890 & 1 & 1.155 & -29.8 & -20.8 & -153.3 \\ \hline
B & 0.888 & 1 & 1.155 & -49.7 & -23.0 & -151.8 \\ \hline
C & 0.889 & 1 & 1.153 & -65.4 & -23.4 & -135.1 \\ \hline\hline
\multicolumn{7}{||c||}{\footnotesize Assumes $r=4\%$ and Gompertz Mortality ($m=88.72,b=10$);} \\ 
\multicolumn{7}{||c||}{\footnotesize $\delta_i$ are given in b.p.; Rates are normalized so $\pi_2=1$} \\ \hline\hline
\end{tabular}\medskip
\caption{Shows the participation rates $\pi_1, \pi_2, \pi_3$ (which are the inverse of the share prices) and corresponding utility loadings $\delta_1,\delta_2,\delta_3$ when there are three cohorts of subscribers in the pool: ages $x_1=60$, $x_2=65$, and $x_3=70$. Utility is logarithmic. Everyone invests 1 dollar ($w_1=w_2=w_3=1$). Tontine designs are: 
\newline A = Natural tontine based on age 65 cohort alone, equitable rates;
\newline B = Natural tontine based on the range of ages, equitable rates;
\newline C = Proportional tontine.
}
\label{table04}
\end{center}
\end{table}
\newpage

\begin{figure}[h]
\begin{center}
\includegraphics[width=0.85\textwidth]{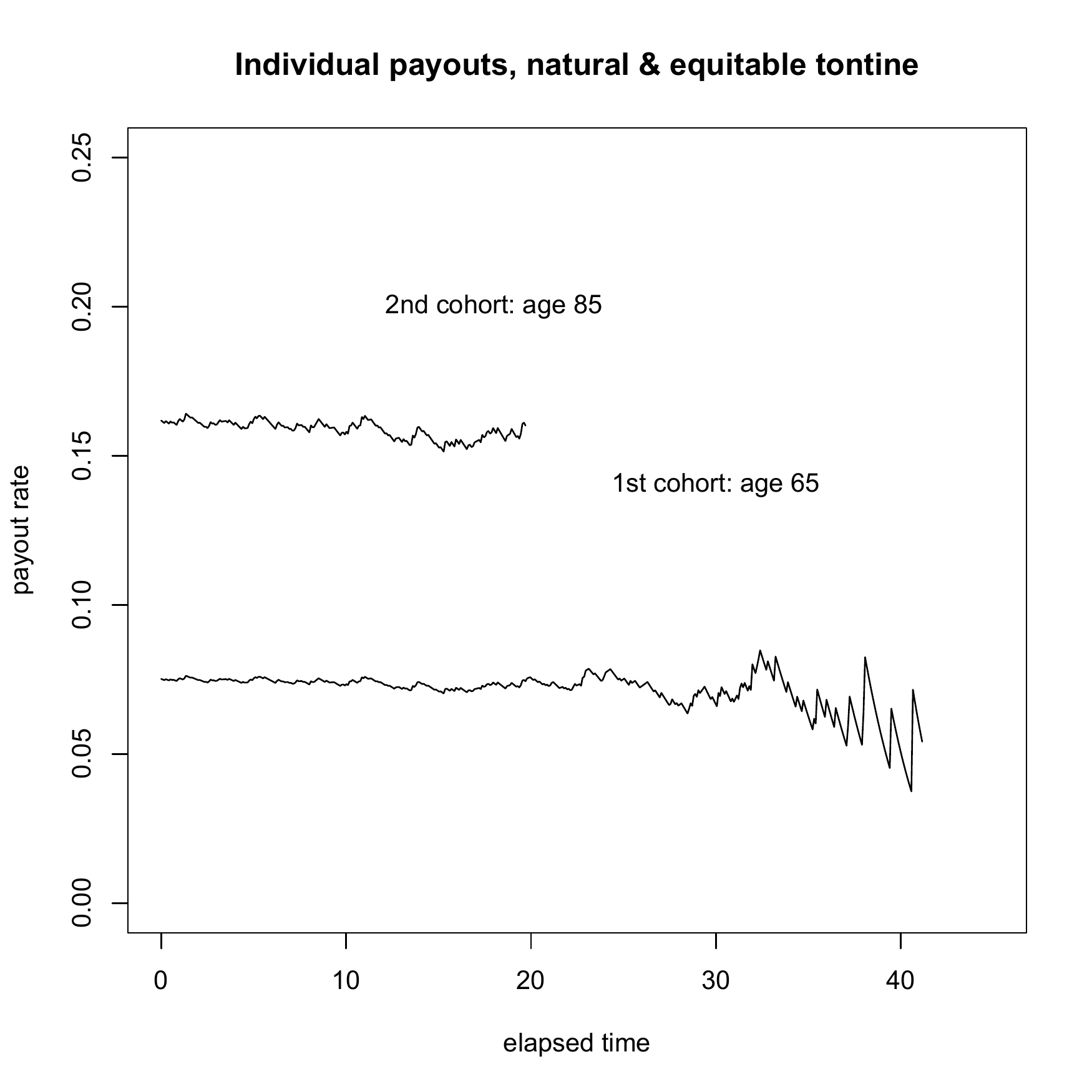} 
\caption{One path of simulated individual payout rates for a natural and equitable tontine with 2 cohorts. First cohort has $n_1=200$ individuals of age $x_1=65$ and second cohort has $n_2=50$ individuals of age $x_2=85$. All individuals invest 1 dollar ($w_1=w_2=1$) but receive an income depending on their group. Simulation assumes Gompertz Mortality ($m=88.72,b=10$) and $r=4\%$.}
\label{fig03}
\end{center}
\end{figure}

\hbox{ }

\newpage

\begin{figure}[h]
\begin{center}
\includegraphics[width=0.4\textwidth]{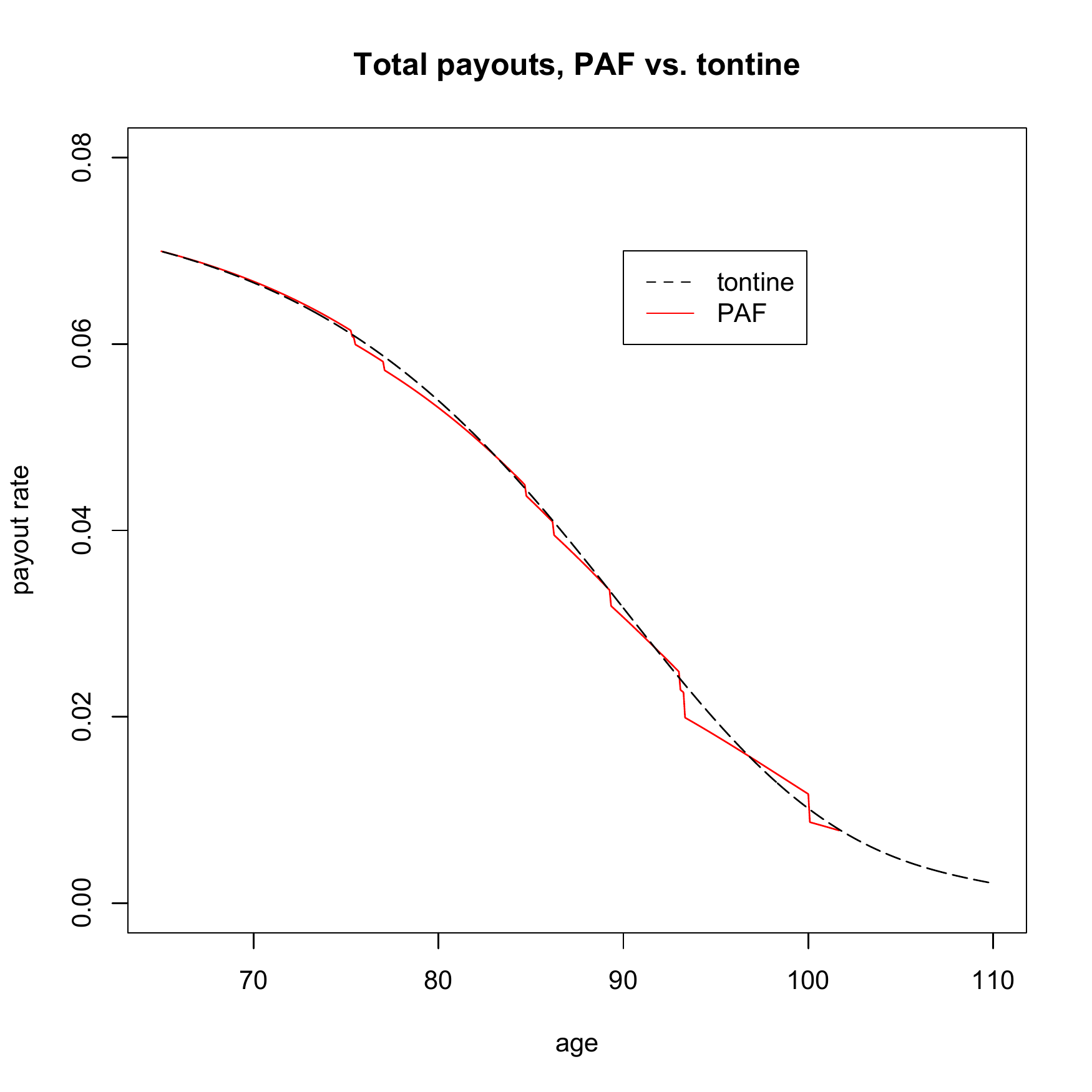} 
\includegraphics[width=0.4\textwidth]{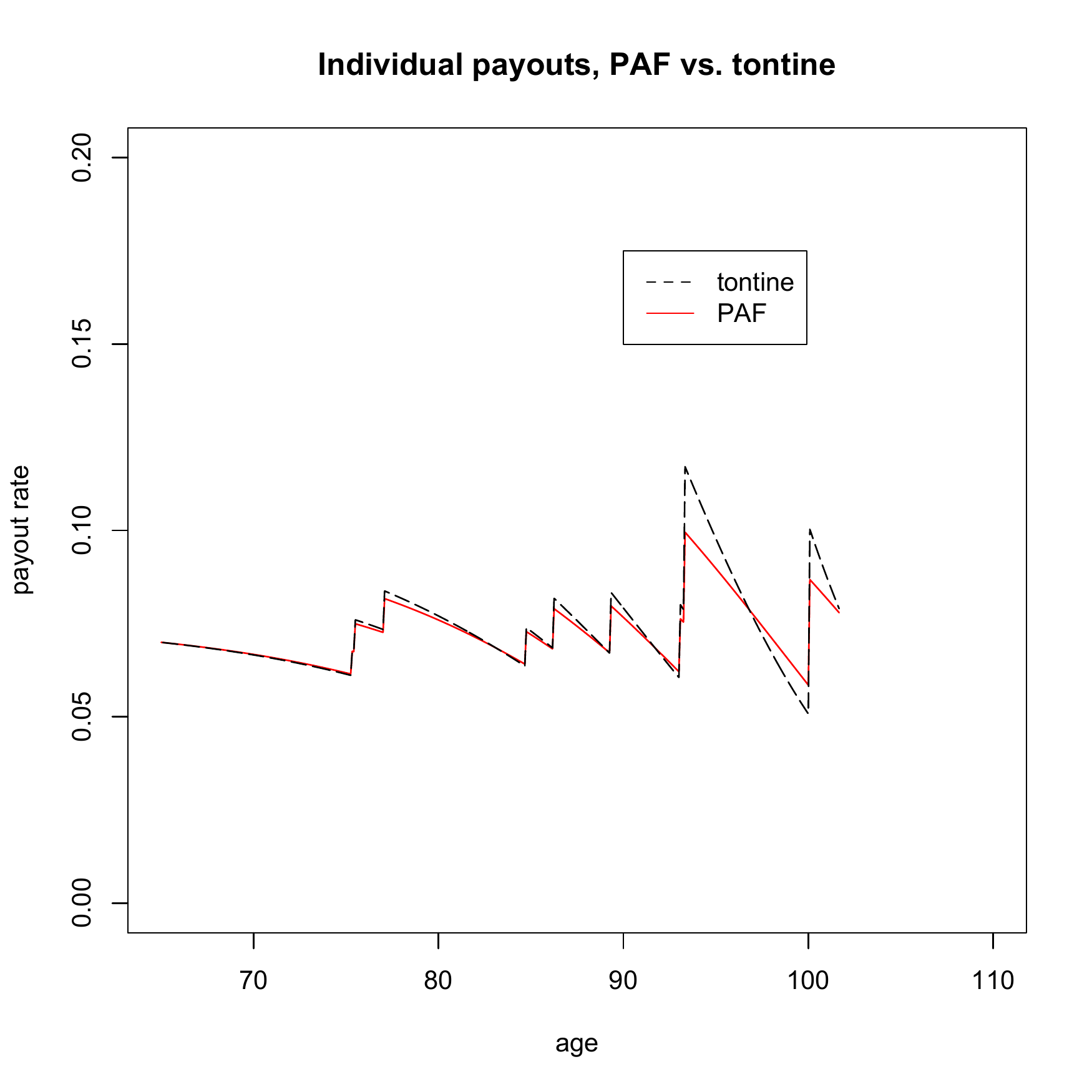} 
\caption{Shows simulated total (left) and individual (right) payout rates for a homogeneous population consisting of 10 individuals age 65. Tontine and PAF are both optimal for risk aversion $\gamma=5$. Total planned tontine payout is shown through age 110; other plots cease upon last death. Assumes Gompertz Mortality ($m=88.72,b=10$) and $r=4\%$.}
\label{fig04}
\end{center}
\end{figure}

\hbox{ }

\newpage

\begin{figure}[h]
\begin{center}
\includegraphics[width=0.85\textwidth]{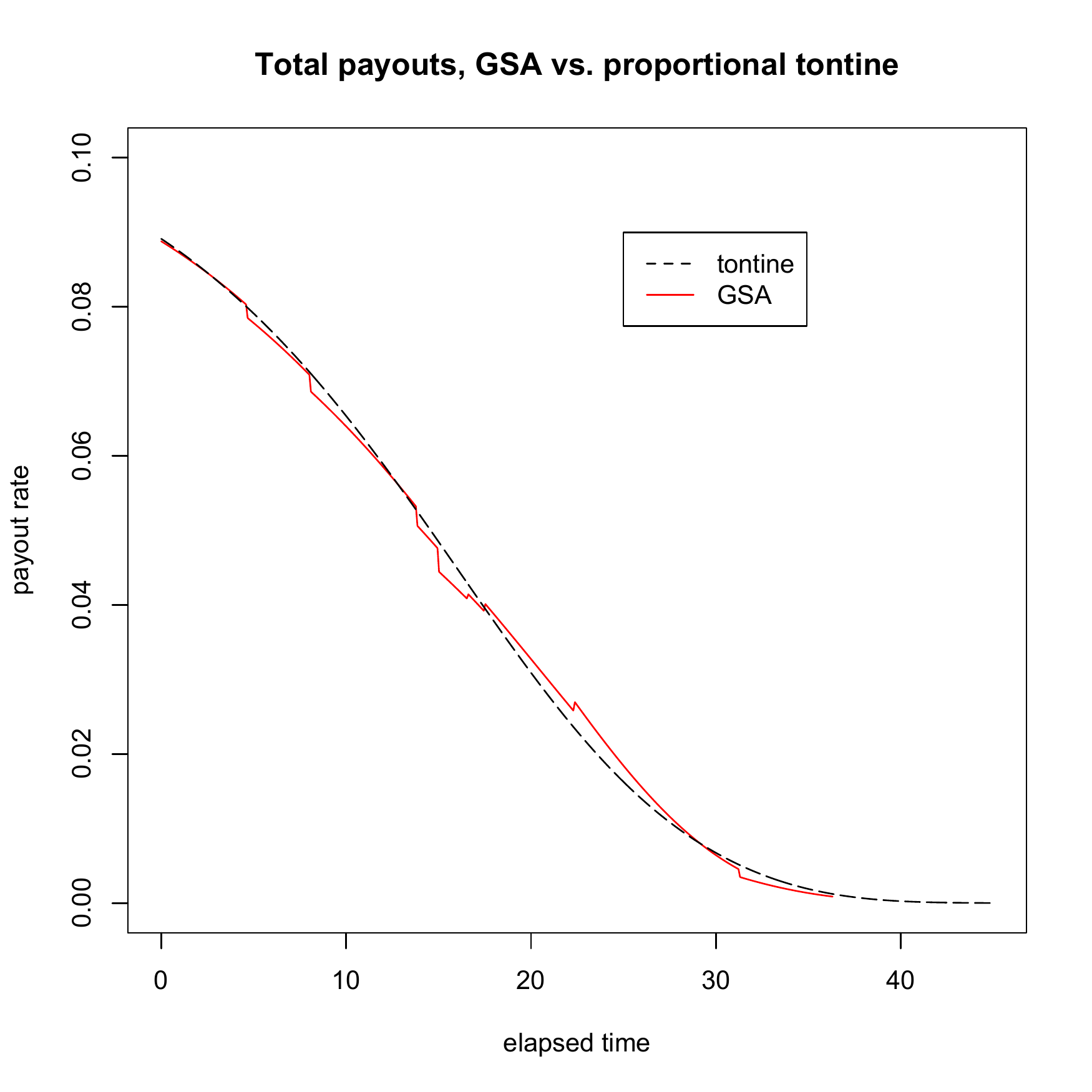} 
\caption{Shows simulated total payout rate for a proportional tontine and a GSA (with monthly payouts). Represents 2 cohorts, ages $x_1=65$ and $x_2=75$, sizes $n_1=n_2=5$, with all individuals investing 1 dollar ($w_1=w_2=1$). Assumes Gompertz Mortality ($m=88.72,b=10$) and $r=4\%$.}
\label{fig05}
\end{center}
\end{figure}

\hbox{ }

\newpage

\section{Appendix: proofs and computations}
\label{sec:appendix}

\subsection{Proof of Theorem \ref{thm:existence}}

We will assume, for now, that each cohort consists of a single individual, i.e. $n=K$, each $n_i=1$, and the $N_i(t)$ are Bernoulli. 
There will be no loss of generality in doing so, as far as the proof of uniqueness goes. To see this, simply split the cohorts up. For existence however, an additional argument will then be needed to establish the sufficiency of \eqref{typecondition1}. In fact, a short proof of necessity could be extracted from the proof of Lemma \ref{individualexistence} below (see formula \eqref{eqn:necessity}). The more complex structure we develop below is needed principally for sufficiency.

Let $\mathcal{P}=\{(\pi_1,\dots,\pi_n): \text{$0< \pi_i\le 1$ for each $i$, and $\max_i \pi_i=1$}\}$. Scaling $\pi$ does not affect the $F_i$, so every value of $F=(F_1,\dots,F_n)$ can be realized with some $\pi\in\mathcal{P}$. We start with the uniqueness question.

\begin{proof}[Proof of (a) of Theorem \ref{thm:existence}]

Suppose $\pi\neq\tilde\pi$ are both equitable, and not multiples of each other. Interpolate between them using $\pi(s)=s\pi+(1-s)\tilde\pi$. Then 
\begin{align}
\frac{d}{ds}F_i(\pi(s))
&=\int_0^\infty e^{-rt}{}_tp_{x_i}wd(t)E_i\Big[\frac{d}{ds}\frac{\pi_i(s) }{\sum_{j} \pi_j(s) w_jN_j(t)}\Big]\,dt \nonumber\\
&=\int_0^\infty e^{-rt}{}_tp_{x_i}wd(t)E_i\Big[\frac{\sum_{j}[\pi_i'(s)\pi_j(s)-\pi_i(s)\pi_j'(s)]w_jN_j(t) }{(\sum_{j} \pi_j(s) w_jN_j(t))^2}\Big]\,dt.
\label{eqn:perturb}
\end{align}
Moreover
\begin{multline*}
\pi_i'(s)\pi_j(s)-\pi_i(s)\pi_j'(s)
=(\pi_i-\tilde\pi_i)\Big[s(\pi_j-\tilde\pi_j)+\tilde \pi_j\Big] -\Big[(s(\pi_i-\tilde\pi_i)+\tilde\pi_i\Big](\pi_j-\tilde\pi_j)\\
=(\pi_i-\tilde\pi_i)\tilde \pi_j-(\pi_j-\tilde\pi_j)\tilde\pi_i
=\pi_i\tilde\pi_j-\pi_j\tilde\pi_i=\pi_i\pi_j\Big(\frac{\tilde\pi_j}{\pi_j}-\frac{\tilde\pi_i}{\pi_i}\Big)
\end{multline*}
for each $s\in[0,1]$. 
In particular, if we choose $i$ to minimize $\frac{\tilde\pi_i}{\pi_i}$, it follows that this expression is $\ge 0$ for every $j$, and $>0$ for some $j$ (since $\tilde\pi$ is not a multiple of $\pi$). Therefore for this choice of $i$ we have that $\frac{d}{ds}F_i(\pi(s))>0$. Thus $F_i(\tilde\pi)>F_i(\pi)$, which is a contradiction. It follows that uniqueness holds.
\end{proof}

To prove existence, and in particular to understand \eqref{typecondition1}, we need to dig deeper into the structure of the optimal $\pi$. In particular, to consider limiting cases when some of the $\pi_i\to 0$. To do that, we borrow an idea from the theory of Martin boundaries (see for example Doob (1984)) and embed $\mathcal{P}$ in a compact set $\mathcal{P}_0$. 

Set $\eta=\{1,\dots,n\}$. For non-empty $A\subset\eta$ and $\pi\in \mathcal{P}$, let $\pi_A=(\frac{\pi_i}{\max_{j\in A}\pi_j})_{i\in A}$. Set $g(\pi)=(\pi_A)_{\emptyset\neq A\subset\eta}$, so $g:\mathcal{P}\to[0,1]^m$, where $m=\sum_{\emptyset\neq A\subset\eta}|A|=\sum_{k=1}^n k\binom{n}{k}=\sum_{j=0}^{n-1}n\binom{n-1}{j}=n2^{n-1}$. Let $\mathcal{P}_0$ be the closure of $g(\mathcal{P})$ in $[0,1]^m$, so $g$ is a continuous embedding of $\mathcal{P}$ in $\mathcal{P}_0$. We will think of $\mathcal{P}$ as a subset of $\mathcal{P}_0$, so to some extent we will use notation that identifies $\pi\in\mathcal{P}$ with $g(\pi)\in\mathcal{P}_0$. In particular, we will freely use $\pi$ to denote either an element of $\mathcal{P}$ or an element of $\mathcal{P}_0$, and in both cases will use the same notation $\pi_A$ for its components, as defined above. 

Let $\pi\in\mathcal{P}_0$. Then $\pi_\eta\in[0,1]^n$ may have some (but not all) of its components $=0$. Writing $\pi_\eta=(\pi_{\eta,1},\dots,\pi_{\eta,n})$, we let $A_0=\{i\mid \pi_{\eta,i}\neq 0\}$. If $A_0\neq\eta$, then $\pi_{\eta\setminus A_0}$ may in turn have some (but not all) of its components $=0$. Let $A_1=\{i\in\eta\setminus A_0\mid \pi_{\eta\setminus A_0,i}\neq 0\}$. Continuing this way with $\pi_{\eta\setminus A_0\cup A_1}$ etc., we partition $\eta$ into a finite number of non-empty subsets $A_0, A_1, \dots,A_J$ such that the components of each $\pi_{A_j}$ are all non-zero. Of course, if $\pi$ actually $\in\mathcal{P}$ then $A_0=\eta$. In fact, $\pi$ may be recovered from these $A_j$ and $\pi_{A_j}$, as the following Lemma shows. 
\begin{lemma}
\label{recursion}
Let $A\subset\eta$ be nonempty, and let $\pi\in\mathcal{P}_0$. Define the $A_j$ as above, and let $i=\min\{j| A\cap A_j\neq\emptyset\}$. Then for $k\in A$,
$$
\pi_{A,k}=
\begin{cases}
\frac{\pi_{A_i,k}}{\max_{j\in A\cap A_i}\pi_{A_i,j}}, & k\in A\cap A_i\\
0, & k\in A\setminus A_i.
\end{cases}
$$
\end{lemma}
\begin{proof}
To see this, choose $\pi^{(m)}\in\mathcal{P}$ such that $g(\pi^{(m)})\to \pi$. Set $B=\cup_{j\ge i} A_j$, so $A\subset B$. For $k\in B$, we have
$$
\pi_{B,k}=\lim_{m\to\infty} \pi^{(m)}_{B,k}=\lim_{m\to\infty} \frac{\pi^{(m)}_k}{\max_{j\in B}\pi^{(m)}_j}.
$$
By definition of $\pi_{B,k}$, this is non-zero precisely for $k\in A_i$, so if $j\in A_i$ and $k\in B\setminus A_i$ then $\frac{\pi^{(m)}_k}{\pi^{(m)}_j}\to 0$. Therefore $\max_{j\in A}\pi^{(m)}_j=\max_{j\in A\cap A_i}\pi^{(m)}_j$ for sufficiently large $m$, and if $k\in A$ then
$$
\pi_{A,k}
=\lim_{m\to\infty}\pi^{(m)}_{A,k}
=\lim_{m\to\infty}\frac{\pi^{(m)}_k}{\max_{j\in A}\pi^{(m)}_j}
=\lim_{m\to\infty}\frac{\pi^{(m)}_k}{\max_{j\in A\cap A_i}\pi^{(m)}_j}.
$$
By the above, this $=0$ if $k\in A\setminus A_i$. If $k\in A\cap A_i$ we may divide numerator and denominator by ${\max_{j\in A_i}\pi^{(m)}_j}$ to see that it 
$$
=\lim_{m\to\infty}\frac{\pi^{(m)}_{A_i,k}}{\max_{j\in A\cap A_i}\pi^{(m)}_{A_i,j}}
=\frac{\pi_{A_i,k}}{\max_{j\in A\cap A_i}\pi_{A_i,j}},
$$ 
as required.
\end{proof}
What is going on in this argument is that for $\pi\in\mathcal{P}_0$ and $A_0,A_1,\dots,A_J$ as above, having a sequence $\pi(n)\in\mathcal{P}$ converge to $\pi$ in the topology of $\mathcal{P}_0$ means that the $\pi^{(m)}_j$ converge to non-zero values for $j\in A_0$, they converge to $0$ at a common rate for $j\in A_1$ (with $\pi_{A_1}$ giving a suitably renormalized limit), they converge to $0$ at a faster rate for $j\in A_2$, etc. 

For a given payout function $d(t)$ we defined a tontine above, corresponding to any $\pi\in\mathcal{P}$. We can generalize this to any $\pi\in\mathcal{P}_0$. It pays only to individuals in $A_0$, as long as any of them survive, using participation rates $\pi_i$. As soon as the last of these individuals dies, it starts paying out to individuals in $A_1$, using participation rates $\pi_{A_1,i}$. Once they all die, it starts paying out to individuals in $A_2$, using rates $\pi_{A_2,i}$, etc. Since payments are contingent on the extinction of an earlier group, we call this generalization a {\it contingent tontine}. 

If there is a contingent tontine that is favourable to the last group to start collecting, it is quite plausible that no $\pi$ can achieve equity. The content of Theorem \ref{thm:existence} is that these two statements are in fact equivalent. Moreover, we shall see in the course of the proof that equation \eqref{typecondition1} precisely captures the failure of the first statement. 

We may now generalize the definition of the present value functions $F_i(\pi)$. Let $\pi\in \mathcal{P}_0$, and suppose that $i\in A_k$. Let $T_0=0$, and for $1\le j\le J+1$ let $T_k$ be the time the last survivor from $A_0\cup\dots\cup A_{k-1}$ dies. Let $\zeta_i$ be the lifetime of individual $i$. Define
\begin{equation}
F_i(\pi)=\int_0^\infty e^{-rt}wd(t)E\Big[\frac{\pi_{A_k,i} }{\sum_{j\in A_k} \pi_{A_k,j} w_jN_j(t)}1_{\{T_k<t<\zeta_i}\}\Big]\,dt.
\label{contingenteqn}
\end{equation}
The point of passing to the more complicated index set $\mathcal{P}_0$ is the following:
\begin{lemma}
Each $F_i:\mathcal{P}_0\to\mathbb{R}$ is continuous.
\label{continuity}
\end{lemma}
\begin{proof}
Let $\pi\in\mathcal{P}_0$. Suppose that $\pi^{(m)}\to \pi$. Assume to start with that each $\pi^{(m)}\in\mathcal{P}$. Define $A_0, \dots ,A_J$ as above (using $\pi$), and likewise $T_0, \dots,T_{J+1}$, and let $i\in A_\ell$. Set $B_k=A_k\cup\cdots\cup A_J$. Then 
\begin{align*}
F_i(\pi^{(m)})
&=\int_0^\infty e^{-rt}wd(t)E\Big[\frac{\pi^{(m)}_{i} }{\sum_{j\in\eta}\pi^{(m)}_j w_jN_j(t)}1_{\{t<\zeta_i}\}\Big]\,dt\\
&=\sum_{k} \int_0^\infty e^{-rt}wd(t)E\Big[\frac{\pi^{(m)}_{i} }{\sum_{j\in\eta} \pi^{(m)}_{j} w_jN_j(t)}1_{\{T_k<t<\zeta_i\land T_{k+1}}\}\Big]\,dt\\
&=\sum_{k\le \ell} \int_0^\infty e^{-rt}wd(t)E\Big[\frac{\pi^{(m)}_{i} }{\sum_{j\in B_k} \pi^{(m)}_{j} w_jN_j(t)}1_{\{T_k<t<\zeta_i\land T_{k+1}}\}\Big]\,dt\\
&=\sum_{k\le \ell} \int_0^\infty e^{-rt}wd(t)E\Big[\frac{\tilde \pi^{(m)}_{i,k} }{\sum_{j\in B_k} \tilde\pi^{(m)}_{j,k} w_jN_j(t)}1_{\{T_k<t<\zeta_i\land T_{k+1}}\}\Big]\,dt.
\end{align*}
where $\tilde \pi^{(m)}_{q,k}=\frac{\pi^{(m)}_q}{\max_{j\in A_k}\pi^{(m)}_j}$. Now send $m\to\infty$. If $j\in A_k$ then $\tilde \pi^{(m)}_{j,k}\to\pi_{A_k,j}$, while if $j\in B_{k+1}$ then $\tilde \pi^{(m)}_{j,k}\to 0$. In particular, dominated convergence implies that the terms with $k<\ell$ vanish in the limit, while the $k=\ell$ term converges to to $F_i(\pi)$. 

The general case now follows. If $\pi^{(m)}\to\pi$ but we no longer assume $\pi^{(m)}\in\mathcal{P}$, simply choose $\tilde\pi^{(m)}\in\mathcal{P}$ such that $|F_i(\pi^{(m)})-F_i(\tilde \pi^{(m)})|\to 0$ and $\|\pi^{(m)}-\tilde\pi^{(m)}\|\to 0$. Then $\tilde \pi^{(m)}\to\pi$ so $\lim F_i(\pi^{(m)})=\lim F_i(\tilde \pi^{(m)})=F_i(\pi)$. 
\end{proof}
Define a {\it very equitable} participation rate to be any $\pi\in\mathcal{P}_0$ that minimizes $\theta_0(\pi)=\max_{i,k} |F_i(\pi)-F_k(\pi)|$ over $\pi\in\mathcal{P}_0$. By Lemma \ref{continuity} and compactness of $\mathcal{P}_0$, such a $\pi$ exists. Of course, $\pi$ is equitable if and only if $\theta_0(\pi)=0$ and $\pi\in\mathcal{P}$. 
\begin{lemma} Let $\pi$ be very equitable. Let $a_0=\min_{i\in\eta}F_i(\pi)$, and  $a_J=\max_{i\in\eta}F_i(\pi)$. Then $F_i(\pi)=a_0$ for every $i\in A_0$, and $F_i(\pi)=a_J$ for every $i\in A_J$.
\label{mostequitable}
\end{lemma}
\begin{proof}
Let $\pi$ be very equitable. 
As in the proof of uniqueness, we will perturb $\pi$ to improve equity. Fix $j$. Let $\tilde A$ consist of the $i\in A_j$ which minimize $F_i(\pi)$ over $A_j$, and set 
$$
\pi_{A_j,i}(s)=
\begin{cases}
\pi_{A_j,i}(1+s), &i\in\tilde A\\
\pi_{A_j,i}, & i\in A_j\setminus\tilde A.
\end{cases}
$$
We do not perturb $\pi_{A_k,q}$ for any $k\neq j$, so there is no impact on $F_q$ for $q\notin A_j$. 
Note that this $\pi(s)$ may not lie $\in\mathcal{P}_0$, as $\max_{i\in A_j}\pi_{A_j,i}(s)$ may now be $\neq 1$. This will not turn out to matter, and could in any case be remedied by rescaling at a suitable point in the argument. 
A calculation as in \eqref{eqn:perturb} shows that for $i\in A_j$,
\begin{multline*}
\frac{d}{ds}F_i(\pi(s))=\\
\int_0^\infty e^{-rt}{}_tp_{x_i}wd(t)E\Big[\frac{\sum_{k\in A_j}[\pi_{A_j,i}'(s)\pi_{A_j,k}(s)-\pi_{A_j,i}(s)\pi_{A_j,k}'(s)]w_kN_k(t) }{(\sum_{k\in A_j} \pi_{A_j,k}(s) w_kN_k(t))^2}1_{\{T_k<t<\zeta_i\land T_{k+1}}\}\Big]\,dt.
\end{multline*}
Assume that $\tilde A\neq A_j$, i.e. that $F_i(\pi)$ is not constant on $A_j$. I claim that 
$$
\frac{d}{ds}F_i(\pi(s))\text{ is }
\begin{cases}
>0, & i\in \tilde A\\
<0, & i\in A_j\setminus \tilde A.
\end{cases}
$$
In other words, this perturbation brings the lowest $F_i$'s up, and the other $F_i$'s down. 

Suppose $i\in \tilde A$. If $k\in \tilde A$ then $\pi_{A_j,i}'(s)\pi_{A_j,k}(s)-\pi_{A_j,i}(s)\pi_{A_j,k}'(s)=\pi_{A_j,i}\pi_{A_j,k}-\pi_{A_j,i}\pi_{A_j,k}=0$. If $k\in A_j\setminus \tilde A$ then 
$\pi_{A_j,i}'(s)\pi_{A_j,k}(s)-\pi_{A_j,i}(s)\pi_{A_j,k}'(s)=\pi_{A_j,i}\pi_{A_j,k}>0$. Therefore the sum is $>0$ so $\frac{d}{ds}F_i(\pi(s))>0$ too. 

Now suppose $i\in A_j\setminus \tilde A$. If $k\in \tilde A$ then $\pi_{A_j,i}'(s)\pi_{A_j,k}(s)-\pi_{A_j,i}(s)\pi_{A_j,k}'(s)=-\pi_{A_j,i}\pi_{A_j,k}<0$. If $k\notin A$ then 
$\pi_{A_j,i}'(s)\pi_{A_j,k}(s)-\pi_{A_j,i}(s)\pi_{A_j,k}'(s)=0$. Therefore the sum is $<0$ so $\frac{d}{ds}F_i(\pi(s))<0$ too. 

The result of this calculation is that for small $s$, our perturbation reduces the variation of $F_i(\pi)$ on $A_j$, unless $i\mapsto F_i(\pi)$ is already constant on $A_j$. 

Turning to the statement of the lemma, there must be a $j$ such that $F_i(\pi)=a_0$ for every $i\in A_j$, otherwise we could perturb so as to raise every $F_i(\pi)$ which equals $a_0$, while lowering or not changing the other $F_k(\pi)$. (If necessary, rescale to keep $\pi\in\mathcal{P}_0$.) This would reduce $\theta_0(\pi)$ which is impossible. By the same perturbation, we may also assume that for any $j$, either $F_i(\pi)=a_0$ for every $i\in A_j$, or $F_i(\pi)>a_0$ for every $i\in A_j$. Let $J_0$ be the set of $j$ of the former type. Our goal is to show that $0\in J_0$. 

Suppose that $j\ge 1$ belongs to $J_0$, but $j-1$ does not. Consider the following perturbation. Combine $A_{j-1}$ and $A_{j}$, by setting 
$$
\pi_{A_{j-1}\cup A_{j},i}(s)=
\begin{cases}
\pi_{A_{j-1},i}, & i\in A_{j-1}\\
s\pi_{A_{j},i}, & i\in A_{j}
\end{cases}
$$
for $s>0$. This will not impact $F_i(\pi(s))$ for $i$ other than $j-1$ or $j$. For $i\in A_{j}$ this has no effect on the expectation in \eqref{contingenteqn} representing payments after time $T_{j}$, but with positive probability it adds a non-zero contribution from the integral over $[T_{j-1},T_{j})$. Therefore $F_i(\pi(s))$ increases for each $i\in A_{j}$. A derivative calculation similar to that given above shows that the $F_i(\pi(s))$ decrease for $i\in A_{j-1}$. 

This perturbation may or may not decrease $\theta_0(\pi)$. But if $0\notin J_0$ then we may apply it in turn to the first $j$ in $J_0$, then the second $j$ in $J_0$, etc., until eventually $\theta_0(\pi)$ will decrease. This would be a contradiction, so it follows that $0\in J_0$. 

We may apply a similar argument to $a_J$ to prove the remaining conclusions.
\end{proof}

We are now ready to prove the existence portion of Theorem \ref{thm:existence}, under our additional restriction that $n=K$. For $A\subset\eta$, let $\alpha_A=\frac{1}{w}\sum_{i\in A}w_i$ be the percentage of the total initial investment contributed by members of $A$.
\begin{lemma}
Fix $d(t)$ as well as the $x_i$ and $w_i$, and assume that each cohort consists of a single individual.  For there to exist a choice of $\pi\in\mathcal{P}$ with $F_i(\pi)=1-\epsilon$ for each $i$, it is necessary and sufficient that for every $A\subset\eta$ with $\emptyset\neq A\neq\eta$ we have 
\begin{equation}
\int_0^\infty e^{-rt}d(t)(\prod_{i\notin A}{}_tq_{x_i})(1-\prod_{i\in A}{}_tq_{x_i})\,dt < \alpha_A(1-\epsilon).
\label{concavity1}
\end{equation}
\label{individualexistence}
\end{lemma}
We may think of \eqref{concavity1} failing for one of two reasons -- the presence of particularly elderly individuals, or of individuals who contribute a disproportionately large fraction of the initial investment. In either case, we let $A$ consist of the remaining people (younger, or investing less). Condition \eqref{concavity1} will fail if either $\prod_{i\notin A}{}_tq_{x_i}$ is large (i.e. the $A^c$ individuals all die early into the tontine), or if $\alpha_A$ is small  (i.e. the $A^c$ individuals over-invest). A well-designed $d(t)$ will attempt to mitigate these possibilities, though we have seen that this is not always possible. In particular, if we fix a choice of $d(t)$ but allow the $w_i$ to vary, there will always be a choice for the $w_i$ that makes one of the $\alpha_A$ small enough to force \eqref{concavity1} to fail. 

\begin{proof}
Assume \eqref{concavity1}, and let $\pi\in\mathcal{P}_0$ be very equitable. Let the $a_J$ be as in Lemma \ref{mostequitable}. By Lemma \ref{mostequitable} we have $a_J\ge 1-\epsilon$, since a suitably weighted average of the $F_i(\pi)$ equals $1-\epsilon$. Suppose that there is no equitable $\pi$, or that the equitable $\pi$ belongs to $\mathcal{P}_0\setminus\mathcal{P}$. Then $J>0$, so $0<|A_J|<n$. Moreover
\begin{align*}
(1-\epsilon) \alpha_{A_J}
&\le \sum_{i\in A_J} \frac{w_i}{w}F_i(\pi)\
=\int_0^\infty e^{-rt}d(t)\sum_{i\in A_J}E\Big[\frac{\pi_{A_J,i}w_i}{\sum_{k\in A_J} \pi_{A_J,k} w_kN_k(t)}1_{\{T_{A_J}\le t<\zeta_i\}}\Big]\,dt\\
&=\int_0^\infty e^{-rt}d(t)E\Big[\sum_{i\in A_J}\frac{\pi_{A_J,i}w_iN_i(t)}{\sum_{k\in A_J} \pi_{A_J,k} w_kN_k(t)}1_{\{N_i(t)\neq 0, T_{A_J}\le t\}}\Big]\,dt\\
&=\int_0^\infty e^{-rt}d(t)P\Big(\sum_{i} N_i(t)>0, T_{A_J}\le t\Big)\,dt\\
&=\int_0^\infty e^{-rt}d(t)(\prod_{i\notin A_J}{}_tq_{x_i})(1-\prod_{i\in A_J}{}_tq_{x_i})\,dt,
\end{align*}
which violates \eqref{concavity1}. 

Conversely, suppose $\pi\in\mathcal{P}$ is equitable, and let $A\subset \eta$ be non-empty and $\neq \eta$. Then as above,
\begin{align}
(1-\epsilon) \alpha_{A}
&= \sum_{i\in A} \frac{w_i}{w}F_i(\pi)
=\int_0^\infty e^{-rt}d(t)E\Big[\sum_{i\in A}\frac{\pi_{i}w_iN_i(t)}{\sum_{k\in \eta} \pi_{k} w_kN_k(t)}1_{\{N_i(t)\neq 0\}}\Big]\,dt
\nonumber \\
&>\int_0^\infty e^{-rt}d(t)P\Big(\sum_{i\in A} N_i(t)>0, \sum_{i\in \eta\setminus A}N_i(t)=0\Big)\,dt,
\label{eqn:necessity}
\end{align}
which shows \eqref{concavity1}.
\end{proof}

The problem with condition \eqref{concavity1} of Lemma \ref{individualexistence} is that it involves checking $2^n-2$ conditions. Condition \eqref{typecondition1} brings this down to a manageable number, provided we have a modest number of cohorts. We therefore now abandon the assumption that $n=K$, in which case recall that 
$\alpha_A$ once more denotes $\frac{1}{w}\sum_{i\in A}n_iw_i$. 

\begin{proof}[Proof of (b) of Theorem \ref{thm:existence}]
Assume condition \eqref{typecondition1}, which we may restate as
\begin{equation}
\int_0^\infty e^{-rt}d(t)\prod_{i\notin A}{}_tq_{x_i}^{n_i}\,dt < \alpha_A(1-\epsilon)+\epsilon.
\label{typecondition2}
\end{equation}
Condition \eqref{concavity1} involves a general collection of subscribers, which we will take to consist of $0\le k_i\le n_i$ individuals from the $i$th group, $i=1, \dots, K$. Stated in this way it becomes that
\begin{equation}
\int_0^\infty e^{-rt}d(t)\prod_{i=1}^K{}_tq_{x_i}^{n_i-k_i}\,dt < (1-\epsilon)\sum_{i=1}^K\frac{k_iw_i}{w}+\epsilon
\label{concavity2}
\end{equation}
for every choice of $0\le k_i\le n_i$ (other than $(0,\dots,0)$ and $(n_1,\dots,n_K)$). Observe that the left hand side of \eqref{concavity2} is a concave function of each $k_i$ (when the others are fixed), while the right hand side varies in an affine way. Observe also that \eqref{typecondition2} is precisely \eqref{concavity2} for the extreme points of $I=[0,n_1]\times\cdots\times[0,n_K]$ other than $(0,\dots,0)$ and $(n_1,\dots,n_K)$. It is easily verified that the same inequality holds at the last two points as well, but with $=$ rather than $<$. 
This is enough to conclude that the strict inequality \eqref{concavity2} holds at all points of $I$ other than $(0,\dots,0)$ and $(n_1,\dots,n_2)$, so Lemma \ref{individualexistence} applies to give an equitable $\pi\in\mathcal{P}$.

Necessity also follows from Lemma \ref{individualexistence}, since \eqref{typecondition1} is a special case of \eqref{concavity1}.
\end{proof}

\subsection{Computational details}

Motivated by the proof of Theorem \ref{thm:existence}, we calculate equitable participation rates by successively raising those $\pi_i$ for which $F_i(\pi)<1-\epsilon$. 
We choose a relaxation rate $\eta<1$, and cycle through the $i$, raising $\pi_i$ by $\eta$ as often as possible while preserving the above criterion. Then repeat the process this time raising by $\eta^2$, then repeat again raising by $\eta^3$, etc. 
For $K=3$ this appears to converge to 5 digits accuracy, using fewer than 100 evaluations of the vector $F(\pi)$. 

Each equitable $(20,40,20)$ entry of Table \ref{table04} took approximately 4 hours running under R, and involved 100 evaluations of the vector $F(\pi)$. Each evaluation required computing two integrals, using Simpson's rule with 400 time steps. Each time step in turn called for a binomial sum of $n_1n_2n_3$ terms. With $K=2$ the computations are faster (25 integrations, with each time step calling for summing $n_1n_2$ terms); for example the equitable $(500, 500)$ entries in Table \ref{table03} each ran in just over 3 hours, despite accounting for an order of magnitude more people. Note that our code could easily be optimized to run faster, eg. by using a more efficient integration method, or by starting with coarser time steps and refining them as the iteration proceeds. For utility computations, we do seem to need the level of accuracy provided by 400 time steps, but the $\pi$'s are less sensitive and could have been found quicker.

\subsection{Proof of Proposition \ref{prop:PAFandtontine}}
Let $v(t,k,w)$ denote the utility an individual derives from a PAF. Scale invariance shows that $e(t,k,w)=\eta(t,k)w$ for some function $\eta$. When $\gamma=1$, it also shows that $v(t,k,w)=a_{x+t}\log w +v(t,k,1)$. 
In the case $\gamma\neq 1$, Stamos (2008) derives an HJB equation for an individual's utility $v(t,k,w)$. The same argument applies when $\gamma=1$ and shows (in our notation) that 
\begin{multline*}
v_t(t,k,w)+\lambda_{x+t}(k-1)[v(t,k-1,w\frac{k}{k-1})-v(t,k,w)]-(r+\lambda_{x+t})v(t,k,w) + \\
+\sup_\eta \Big[\log(\eta w)+wv_w(t,k,w)(r-\eta)\Big]=0
\end{multline*}
for $k\ge 2$ (with a similar equation when $k=1$, except without the second term. Optimizing over $\eta$, and substituting our expression for $v$ gives that $a_{x+t}=wv_w(t,k,w)=1/\eta$. Therefore $\overline{e}(t,k,\overline w)=\frac{\overline{w}}{na_{x+t}}$, which is independent of $k$. Therefore $\overline{W}_t$ evolves according to the differential equation
$$
d\overline{W}_t=\Big[r-\frac{1}{a_{x+t}}\Big]\overline{W}_t\,dt.
$$
It is easily checked that the solution (with initial condition $\overline{W}_0=n$) is $\overline{W}_t=\frac{na_{x+t}\cdot{}_tp_x}{a_x}$, from which we conclude that 
$\overline{e}(t,N_t,\overline{W}_t)=\frac{1}{a_x}{}_tp_x$, as required. 

\end{document}